\documentclass[sigconf,screen,nonacm]{acmart}
\setcopyright{none}
\acmDOI{}
\acmISBN{}
\acmPrice{}
\renewcommand\footnotetextcopyrightpermission[1]{}
\usepackage[T1]{fontenc}
\usepackage{amsmath,amsthm,mathtools,booktabs,enumitem,array,graphicx,makecell}
\usepackage{tikz}
\usetikzlibrary{arrows.meta}
\newtheorem{theorem}{Theorem}[section]
\newtheorem{lemma}[theorem]{Lemma}
\newtheorem{corollary}[theorem]{Corollary}
\newtheorem{definition}[theorem]{Definition}
\newtheorem{proposition}[theorem]{Proposition}
\newtheorem{remark}[theorem]{Remark}
\numberwithin{equation}{section}

\newcommand{\dist}{\operatorname{dist}}
\newcommand{\poly}{\operatorname{poly}}
\newcommand{\one}{\mathbf 1}

\newcommand{\dr}{\delta_{\rm resp}}
\newcommand{\es}{\varepsilon_{\rm cov}}
\newcommand{\ep}{\varepsilon_{\rm orient}}
\newcommand{\pl}[1]{\mathsf{#1}}
\newcommand{\alphaor}{\alpha_{\rm or}}
\newif\ifregularextension
\regularextensionfalse
\begin{document}
\title[Invariant Simplices and Shared Probabilistic Automata]{Local Optimality Does Not Compose: Invariant Simplices and the Complexity of Shared Probabilistic Automata}
\author{Yixin Zhao}
\authornote{Corresponding author.}
\affiliation{%
  \institution{Beijing Key Laboratory of Fault-Tolerant Quantum Computing, Beijing Academy of Quantum Information Sciences}
  \city{Beijing}
  \postcode{100193}
  \country{China}}
\affiliation{%
  \institution{Beijing National Laboratory for Condensed Matter Physics, Institute of Physics, Chinese Academy of Sciences}
  \city{Beijing}
  \postcode{100190}
  \country{China}}
\affiliation{%
  \institution{University of Chinese Academy of Sciences}
  \city{Beijing}
  \postcode{100190}
  \country{China}}
\email{zhaoyixin22@mails.ucas.ac.cn}
\begin{abstract}
Probabilistic models are often constructed separately for different test suites
and then expected to share a single latent dynamical system. We study bounded-state
shared realization for probabilistic automata from finite partial trace
specifications. A shared realization uses one normalized automaton with named
initial distributions, a common family of row-stochastic letter-transition
matrices, and one final output vector. We compare it with separate local
realizations and a labelled query-specific static relaxation.

Local and labelled query-specific static optima need not compose. For every
$k,L\ge1$, we construct finite partial specifications for which each individual
specification and the labelled query-specific static relaxation have minimum
state count $k$, whereas every exact shared realization requires exactly
$k(L+1)$ states. The factor $L+1$ reflects the prefix progress that must be
tracked by a single shared transition family.

In the linear-rank-tight calibrated regime, the existence of an exact
$k$-state shared realization is equivalent to finding a $(k-1)$-simplex
containing the root responses, contained in the response polytope, and
forward-invariant under every single-letter affine update.

For finite rational specifications, exact shared realization with at most $K$
states is $\exists\mathbb{R}$-complete, while distinguishing zero minimax
defect from inverse-polynomial defect is $\mathsf{PromiseNP}$-complete. Both
classifications hold over a fixed five-letter alphabet when the local and
labelled query-specific static optima already equal $K$. They remain valid
with Boolean final output vectors, and the robust gap persists over the full
cube $[0,1]^d$. Thus the hardness arises from shared word composition.
\end{abstract}
\maketitle
\hypersetup{pdftitle={Local Optimality Does Not Compose: Invariant Simplices and the Complexity of Shared Probabilistic Automata},pdfauthor={Yixin Zhao}}
\section{Introduction}
\label{sec:introduction}

Finite sets of bounded executions often specify a probabilistic reactive
system only partially.  A probabilistic automaton maps a controlled word to an
acceptance probability by evolving an initial distribution through stochastic
transition matrices and applying a final output vector
~\citep{rabin1963probabilistic,paz1971probabilistic}. We represent selected
observations by a finite partial trace specification: a declaration
$(h,u)\mapsto y$ prescribes the expected final output after initial
distribution $h$ and input word $u$. For $u=a_1\cdots a_t$, the bounded
action-indexed trace atom
$\langle u\rangle E$ has value
$\operatorname{val}_h(\langle u\rangle E):=p_hM_{a_1}\cdots M_{a_t}e$, where
$E$ denotes the final quantitative observable. A specification is therefore a
finite conjunction of equalities, or their common-$\epsilon$ relaxations, on
such trace values. When $e=\mathbf 1_F$ is Boolean, the atom is exactly the
probability that the controlled trace $u$ terminates in $F$. Realizing a
specification is finite-state probabilistic synthesis under a prescribed state
budget.

We work with standard probabilistic automata and study their finite
partial-realization problem in a setting where several named initial
distributions must share a single transition family and one final output
vector. These initial distributions are externally selectable reset inputs
supplied by the test harness; they are not part of the automaton's state
space or transition structure. The one-initial-distribution variant is
discussed in Section~\ref{sec:lics-discussion}. Throughout, ``shared'' means
that all specifications are realized by the same automaton, with each word
interpreted as the ordered product of the common letter-transition matrices.

Several finite partial specifications represent separately calibrated test suites
deployed jointly. Labels record provenance; when overlap targets agree, dropping
labels preserves shared-realization feasibility and minimax defect. Labels remain
only in the baselines: a local realization chooses one automaton per
specification, whereas the query-specific static relaxation shares roots but
gives each labelled query an independent output vector. Our hard instances have
consistent overlaps and query-specific static optimum $K$, while shared
realization uses one carrier, transition family, and output vector for all words.

For a state budget $K$, the unknown initial distributions, stochastic
transitions, and final output vector satisfy the polynomial constraints
$p_hM_ue=y$. Reusing the same matrices $M_a$ across all words imposes the
word-composition constraint missing from the local and query-specific static
baselines. These trace constraints therefore induce an existential
first-order sentence over the real closed field $\mathbb R$, whose quantified
variables describe the implementation. Exact shared realization is the
corresponding bounded-state satisfiability problem.

To witness this separation, we construct a family in which each finite partial
specification has minimum local state count $k$, while the combined
declarations admit a $k$-dimensional query-specific static factorization.
Under one common evolution, however, the model must track both content and
prefix progress, producing an additional factor that grows with $L$.

Our main results classify the computational complexity of shared realization.
Given finite rational partial specifications and a state budget $K$, deciding
whether an exact shared probabilistic automaton with at most $K$ states exists
is $\exists\mathbb{R}$-complete. This remains true over a fixed five-letter
alphabet, on instances whose local and query-specific static optima already
equal $K$, and in the Boolean final-output-vector variant. The robust promise
problem of distinguishing zero minimax defect from defect at least an
inverse-polynomial amount is $\mathsf{PromiseNP}$-complete. Its soundness gap
persists even when the final output vector ranges over the full cube
$[0,1]^d$.

\paragraph{Contributions.}
Our contributions are as follows.
\begin{enumerate}[label=(\arabic*)]
\item \emph{A state separation.} For every $k,L\ge1$, we prove an exact
multiplicative gap: the per-specification local and query-specific static
requirements have dimension $k$, whereas the shared realization requires
exactly $k(L+1)$ states.

\item \emph{A geometric characterization via invariant simplices.}  In the
linear-rank-tight calibrated regime,
 the existence of an exact $k$-state shared realization is equivalent to a
 $(k-1)$-simplex that contains the root response vectors, lies in the response
 polytope, and is forward-invariant under every single-letter affine update.
 Dropping invariance gives the static factorization.

\item \emph{Exact complexity for finite partial specifications.} We prove that deciding
whether rational finite partial specifications admit an exact shared
probabilistic automaton with at most $K$ states is
$\exists\mathbb{R}$-complete, with a fixed five-letter alphabet, one Boolean
final vector in the restricted variant, and exact local and query-specific
static optima equal to $K$.

\item \emph{Robust complexity for finite partial specifications.} For the promise problem separating
zero defect from inverse-polynomial defect, we prove
$\mathsf{PromiseNP}$-completeness under the same structural restrictions.  The
soundness proof combines approximate basis identification, stochastic contraction,
and Boolean rounding.
\end{enumerate}

This separation motivates the rank-tight invariant-simplex interpretation.
The finite partial specification results below instantiate the same
containment-versus-invariance distinction through a different notion of
tightness: exact saturation of the nonnegative-factorization dimension in
every calibrated cell.

\paragraph{Proof overview.}
The exact and robust lower bounds use one calibration mechanism. An identity
table exposes a latent stochastic basis; a five-letter parser and boundary
probes certify reset sectors, while basis queries recover source maps on one
synchronized basis. The exact branch compiles ETR--INV into a partially
specified stochastic map. The robust branch compiles bounded-occurrence
$3$-SAT, using Booleanity schedules and total-variation contraction to control
constant-length error propagation. Both branches share this calibrated
interface and have explicit quantitative bounds.

\paragraph{Relation to prior work.}
Classical stochastic realization starts from complete string-indexed
behaviors~\citep{carlyle1971realizations,paz1971probabilistic}. Equivalence of
fully specified probabilistic automata is polynomial-time
~\citep{tzeng1992equivalence}, whereas minimization from an explicit automaton
is NP-hard and reducible to the existential theory of the reals
~\citep{kiefer2014probabilistic}; both require agreement on all words. For
finite data, prescribed-dimension multiplicity word-automaton consistency is
logspace equivalent over $\mathbb Q$ to the existential first-order theory of
$\mathbb Q$~\citep{kiefer2017multiplicity}, but permits unrestricted linear
representations. Shared stochastic realization instead reconstructs normalized
stochastic dynamics from finite partial observations, with several named
initial distributions sharing one transition semigroup and final output vector.

Partial realization for hidden Markov models uses complete fixed-length block
statistics of one stationary process and supports exact and generic
minimal-realization results~\citep{vidyasagar2006partial,huang2016minimal}.
Here, sparse scalar responses at externally chosen control words and several
named initial distributions define a bounded-state feasibility problem for one
shared row-stochastic transition family and final output vector.

Related $\exists\mathbb R$-complete problems choose stochastic dynamics inside
a fixed ambient structure: memoryless strategies in labelled MDPs can enforce
trace equivalence~\citep{kiefer2020comparing}, and probabilities on a fixed LMC
support graph can enforce agreement on all traces~\citep{fatmi2026perturbation}.
Our problem instead reconstructs a bounded latent carrier from finite
declarations. The hard instances satisfy $C_{\rm loc}=C_{\rm stat}=K$, so the
remaining question is composition under one shared transition family.

The query-specific static relaxation is related to nonnegative and stochastic
factorization: stochastic factorizations correspond to simplices between
nested polytopes~\citep{mond2003stochastic}, while NMF and related completion
problems have universality and nested-polytope formulations
~\citep{shitov2016universality,dobbins2019universality,chistikov2016restricted,kubjas2026completion}.
These static formulations impose no transition-semigroup relation among query
columns. Shared realization adds $e_u=M_ue$ and $M_{uv}=M_uM_v$ for one
stochastic transition family. Theorem~\ref{thm:rank-tight-closure} expresses
this as simultaneous simplex invariance; positive realization and predictive
state representations provide related invariant-cone viewpoints
~\citep{benvenuti2004tutorial,littman2001predictive,boots2011closing}.

\paragraph{Organization.}
Section~\ref{sec:controlled-models} defines the model and decision problems;
Sections~\ref{sec:payload-delay} and~\ref{sec:rank-tight-geometry} give the
state separation and invariant-simplex interpretation. Section
~\ref{sec:lics-finite-partial-specification} presents the finite partial
specification classifications, with detailed reductions and quantitative
estimates in Appendix~\ref{sec:finite-partial-specification-appendix}.

\section{Shared stochastic realizations from finite partial specifications}
\label{sec:controlled-models}

\subsection{Model and three state dimensions}

Let $\Gamma$ be a finite control alphabet and let $H$ be a finite root set.
A row probability vector of dimension $d$ belongs to the probability simplex
$\Delta_d:=\{p\in\mathbb{R}_{\ge 0}^d : \sum_{j=1}^d p_j=1\}$.
A normalized $d$-state realization is a tuple
$\mathcal R=\bigl((p_h)_{h\in H},(M_a)_{a\in\Gamma},e\bigr)$,
where each named initial distribution $h$ specifies
$p_h\in\Delta_d$, each control operator
$M_a\in\mathbb{R}^{d\times d}$ is row-stochastic, and the final output vector
(or effect vector) is $e\in[0,1]^d$ (an acceptance vector in the Boolean case).
Thus the
underlying object is a standard probabilistic automaton with several named
initial distributions; the shared-realization condition requires one common
transition family and one final output vector.
The root set $H$ is part of the finite input; the distributions $p_h$ are
existential implementation variables and do not count as additional states.
For any control sequence
$u=a_1\cdots a_t\in\Gamma^*$, let $M_u:=M_{a_1}\cdots M_{a_t}$ with
$M_\varepsilon:=I_d$. The induced response is given by
\[
R_{\mathcal R}(h,u):=p_hM_ue.
\]
All matrix products adhere to this row-stochastic convention.

Let $\mathfrak M=\{\mathcal U_\lambda : \lambda\in\Lambda\}$ be a nonempty family of
finite partial trace specifications over the shared root set $H$. Each
specification $\mathcal U_\lambda$
specifies a set of declared root--word pairs $D_\lambda\subseteq H\times\Gamma^*$
associated with target response values $y_\lambda(h,u)\in[0,1]$. For each
$\lambda\in\Lambda$, write
$W_\lambda:=\{u\in\Gamma^*:\exists h\in H\text{ with }(h,u)\in D_\lambda\}$.
For a tolerance $\epsilon\ge0$, write
$\mathcal R\models_\epsilon\mathcal U_\lambda$ when
$|R_{\mathcal R}(h,u)-y_\lambda(h,u)|\le\epsilon$ on $D_\lambda$.
We call a normalized $d$-state realization $\mathcal R$ satisfying
$\mathcal R\models_\epsilon\mathcal U_\lambda$ for every $\lambda\in\Lambda$
a $d$-state \emph{shared realization} of $\mathfrak M$ at tolerance
$\epsilon$. In the remainder, ``shared realization'' denotes this
shared probabilistic-automaton model. A shared realization at tolerance
$0$ is called exact.
For shared roots $p_h$ and labelled effects $e_{\lambda,u}$, write
$(p,e)\models_\epsilon^{\rm stat}\mathfrak M$ for the analogous inequalities
$|p_he_{\lambda,u}-y_\lambda(h,u)|\le\epsilon$ on every $D_\lambda$.
Each declaration is indexed by a labelled experiment $(\lambda,h,u)$.  Conflicting
overlap targets remain distinct constraints on the common response
 $R_{\mathcal R}(h,u)$.  If overlap targets agree, the union table obtained
 after dropping specification labels is a single partial response table, and
 deleting duplicate labels changes neither
shared-realization feasibility nor the maximum pointwise error. All hardness
 outputs have this property: labels organize the local and query-specific
baselines without strengthening the shared word-composition problem. For fixed
dimension $d$, existence of a shared realization at tolerance $\epsilon$ is expressible by
existential polynomial equalities and inequalities over $\mathbb R$:
nonnegativity and normalization constrain the stochastic data, while the shared
products $M_u$ impose a global semialgebraic compatibility condition across the
finite declarations.
The three state complexities are
\begin{align*}
C_{\mathrm{loc},\epsilon}(\mathfrak M)
&:=\sup_{\lambda\in\Lambda}\min\bigl\{d\ge1:
   \exists\text{ normalized $d$-state }\mathcal R_\lambda,\\[-0.2em]
&\hspace{11em}\mathcal R_\lambda\models_\epsilon\mathcal U_\lambda\bigr\},\\
C_{\mathrm{stat},\epsilon}(\mathfrak M)
&:=\min\bigl\{d\ge1:
   \exists\,p_h\in\Delta_d,\ e_{\lambda,u}\in[0,1]^d,\\[-0.2em]
&\hspace{8em}(p,e)\models_\epsilon^{\rm stat}\mathfrak M\bigr\},\\
C_{\mathrm{sh},\epsilon}(\mathfrak M)
&:=\min\bigl\{d\ge1:
   \exists\text{ normalized $d$-state }\mathcal R,\\[-0.2em]
&\hspace{9em}
   \mathcal R\models_\epsilon\mathcal U_\lambda\
   (\lambda\in\Lambda)\bigr\}.
\end{align*}

The local complexity optimizes a separate dynamic model for each specification.
The ``static'' relaxation shares roots but gives each labelled query
$(\lambda,u)$ an independent final output vector, whereas shared realization
uses one transition family and one final output vector. An unlabelled benchmark
would additionally identify these vectors across labels. This distinction is
substantive because local witnesses may use different source parameters, so
their pulled-back vectors need not coincide even when scalar responses agree.
Consistent overlaps therefore remove conflicting observations without
identifying latent output vectors. We reserve $C_{\rm stat}$ for the labelled
model: dropping labels preserves shared-realization feasibility and defect,
while both baselines retain the specification partition. We adopt
the convention
$\min\emptyset:=+\infty$. Because the state dimensions are integer-valued, the
supremum in $C_{\mathrm{loc},\epsilon}$ is attained as a maximum whenever the local
state counts admit a uniform finite upper bound.

\begin{table*}[t]
\centering
\caption{The three realization notions used in the paper.}
\label{tab:realization-comparison}
\renewcommand{\arraystretch}{1.04}
\begin{tabular}{@{}p{0.13\textwidth}p{0.18\textwidth}p{0.18\textwidth}p{0.22\textwidth}p{0.18\textwidth}@{}}
\toprule
Notion & Initial distributions & Transition matrices & Final output vectors & Word composition \\
\midrule
Local & Separate per spec. & Separate per spec. & Separate per spec. & Enforced only locally \\
Query-specific static (labelled) & Shared roots & Not used & Independent for each labelled query & Not enforced \\
Shared realization & Shared & One common family & One common vector & Enforced globally \\
\bottomrule
\end{tabular}
\end{table*}

\begin{definition}[Shared state dimension]
\label{def:crc-quantity}
The \emph{shared state dimension} of a specification family
$\mathfrak M$ at tolerance $\epsilon$ is the minimum state count of a shared
realization for $\mathfrak M$ at tolerance $\epsilon$, namely
$C_{\mathrm{sh},\epsilon}(\mathfrak M)$. Here $\epsilon$ is the common
pointwise tolerance over all declared responses: the shared realization must
satisfy every specification within this tolerance. The
subscript ``sh'' records one shared transition family and final output vector. The
associated \emph{shared-realization problem for probabilistic automata} asks, for
a given family $\mathfrak M$,
tolerance $\epsilon$, and state budget $K$, whether $\mathfrak M$ admits a shared realization
with at most $K$ states at tolerance $\epsilon$, equivalently whether
$C_{\mathrm{sh},\epsilon}(\mathfrak M)\le K$.
The quantities $C_{\mathrm{loc},\epsilon}$ and $C_{\mathrm{stat},\epsilon}$ serve as
the corresponding local and query-specific baseline dimensions.
\end{definition}

For exact realization ($\epsilon=0$), we write
\[
C_{\mathrm{loc}}:=C_{\mathrm{loc},0},
\qquad
C_{\mathrm{stat}}:=C_{\mathrm{stat},0},
\qquad
C_{\mathrm{sh}}:=C_{\mathrm{sh},0}.
\]
By construction,
\[
C_{\mathrm{loc}}\le C_{\mathrm{sh}}
\qquad\text{and}\qquad
C_{\mathrm{stat}}\le C_{\mathrm{sh}},
\]
whereas $C_{\mathrm{loc}}$ and $C_{\mathrm{stat}}$ are in general incomparable. The
hard families constructed in this work satisfy
$C_{\mathrm{loc}}=C_{\mathrm{stat}}=K$, and their endpoint calibration
saturates the prescribed budget $K$.  An empty feasible set is assigned value
$+\infty$ for each complexity measure.

\begin{definition}[Minimax defect]
\label{def:crc}
For a finite partial specification instance
\[
 I=\bigl(H,\Gamma,\Lambda,(D_\lambda,y_\lambda)_{\lambda\in\Lambda}\bigr),
\]
put
\[
\begin{aligned}
 D^{\rm tag}(I)&:=
 \{(\lambda,h,u):\lambda\in\Lambda,\ (h,u)\in D_\lambda\},\\
 y_I(\lambda,h,u)&:=y_\lambda(h,u).
\end{aligned}
\]
For a budget $K\in\mathbb Z_{\ge1}$, the \emph{minimax defect} is
\[
J_K(I):=
\min_{\substack{1\le d\le K\\ \mathcal R\ \text{normalized}}}
\max_{(\lambda,h,u)\in D^{\rm tag}(I)}
\bigl|R_{\mathcal R}(h,u)-y_I(\lambda,h,u)\bigr|,
\]
with the conventions $\min\emptyset:=+\infty$ and $\max\emptyset:=0$.
Hence inconsistent targets in
overlapping specifications remain distinct constraints on the same response. For the
instance $I_{\mathfrak M}$ explicitly listing a finite partial specification family and every
$\epsilon\ge0$,
\[
C_{\mathrm{sh},\epsilon}(\mathfrak M)\le K
\quad\Longleftrightarrow\quad
J_K(I_{\mathfrak M})\le\epsilon.
\]
In particular,
$C_{\mathrm{sh},0}(\mathfrak M)\le K$ if and only if
$J_K(I_{\mathfrak M})=0$.
\end{definition}

The \emph{Boolean-final-vector} variant minimizes over realizations whose unique
shared final output vector satisfies $e\in\{0,1\}^d$; the input format is unchanged.
Write $J_K^{\rm B}(I)$ for this minimum.  Then
$J_K(I)\le J_K^{\rm B}(I)$.  The reductions produce Boolean-final-vector exact
witnesses and establish their NO-case lower bounds over the full effect domain
$[0,1]^d$.

The minimax defect provides uniform pointwise control over every labelled declared
experiment. For each fixed dimension $d$, the candidate root distributions,
row-stochastic transition matrices, and final output vector form a compact polytope.
If $D^{\rm tag}(I)=\emptyset$, the objective is identically zero by convention.
Otherwise, because $D^{\rm tag}(I)$ is finite and each response
$R_{\mathcal R}(h,u)$ depends polynomially
on these parameters, the maximum response error is continuous. Taking the minimum over
the finitely many discrete dimensions $1\le d\le K$ therefore attains $J_K(I)$. In
particular, any sequence of realizations within budget $K$ whose defects converge to
zero possesses a subsequence of a fixed dimension converging to an exact normalized
realization.

\subsection{Decision problems and conventions}

\label{sec:conventions}

The class $\exists\mathbb{R}$ consists of the problems polynomial-time
many-one reducible to feasibility of polynomial equalities and inequalities
over the reals with rational coefficients
~\citep{canny1988pspace,renegar1992real,basu2006algorithms,
schaefer2026compendium}.  It satisfies
$\mathsf{NP}\subseteq\exists\mathbb{R}\subseteq\mathsf{PSPACE}$.  A promise
problem belongs to $\mathsf{PromiseNP}$ when YES instances have
polynomial-size certificates accepted by a polynomial-time verifier and no
certificate is accepted on a promised NO instance.  A many-one reduction
between promise problems maps YES instances to YES instances and NO instances
to NO instances; its behavior outside the promise is unrestricted.

Our finite inputs explicitly list rational targets for a finite set of labelled
triples $(\lambda,h,u)$, together with the alphabet and a state budget $K$.
The budget $K$ is unary encoded or polynomially bounded in the remaining input
length. The exact problem asks whether $J_K(I)=0$. The robust problem takes a
rational parameter $0<\epsilon\le1$ of polynomial encoding length and
distinguishes $J_K(I)=0$ from $J_K(I)\ge\epsilon$; all robust reductions
preserve this zero-versus-$\epsilon$ gap.

We write $\varepsilon$ for the empty word and $\dr$ for the pointwise defect
of a candidate realization. We compare stochastic rows in $\ell_1$, which is
twice their total-variation distance. The point mass at a latent state $z$ is
denoted by $\delta_z$. Parser letters are typeset as
$\pl0,\pl1,\#,\pl P,\pl T$; italic $P$ and $T$ denote a polytope and a
simplex, respectively.

\section{State separation under shared word composition}
\label{sec:payload-delay}
\label{sec:payload-full-proof}

This section isolates the state dimension required by shared word composition.
Payload probes distinguish the root states at each prefix length, while prefix
length probes identify the stopping point independently of the active root. The
exact $0/1$ responses enforce disjoint support constraints across both
coordinates, forcing the latent state space to realize the full product
structure.

\paragraph{A finite example.}
The case $k=2,L=1$ already displays the separation.  Take roots $0,1$.
Two payload specifications query $c,bc$ and $ac,abc$, respectively, with the roots
receiving complementary one-hot answers; two prefix-length specifications query
$\varepsilon$ with target zero and $a$ with target one. Each specification has a
  two-state realization, and independent static query effects fit all four
  specifications on the same two-state carrier. A shared transition system
  records both the root bit and whether $a$ has occurred, giving shared state
  dimension four. The theorem below amplifies these two coordinates to
  dimensions $k$ and $L+1$.

\begin{theorem}[Shared-state separation]
\label{thm:main-product-separation}
For every $k\ge1$ and $L\ge1$, over the root set
$H=\mathbb{Z}_k$ and control alphabet $\Gamma=\{a,b,c\}$, optionally padded
by two inert self-loop letters to embed it into the five-letter interface used
below, define the specification family
\[
\begin{aligned}
\mathfrak{M}_{k,L}&:=
 \{\mathcal U_t^{\rm pay},\mathcal U_t^{\rm del}:0\le t\le L\},\\
D_t^{\rm pay}&:=H\times\{a^t b^j c:0\le j<k\},
&D_t^{\rm del}&:=H\times\{a^t\}.
\end{aligned}
\]
For each root $i\in\mathbb{Z}_k$, prescribe the target responses
$y_t^{\rm pay}(i,a^t b^j c):=\mathbf{1}[\,i+j\equiv0\pmod{k}\,]$ on
$D_t^{\rm pay}$ and
$y_t^{\rm del}(i,a^t):=\mathbf{1}[\,t=L\,]$ on $D_t^{\rm del}$. Then the normalized realization
complexities satisfy
\[
C_{\mathrm{loc}}(\mathfrak{M}_{k,L})
=
C_{\mathrm{stat}}(\mathfrak{M}_{k,L})
=
k,
\qquad
C_{\mathrm{sh}}(\mathfrak{M}_{k,L})=k(L+1).
\]
 Because the declarations agree on every overlap, the union table obtained after
 dropping specification labels is well defined, and the same static equality holds for this unlabelled
table with one effect per word.
Consequently, shared word composition incurs an additive overhead of
$kL$ states and an unbounded multiplicative separation by a factor of $L+1$
over the local and query-specific baselines, with a fixed alphabet and
a single Boolean final output vector.
\end{theorem}

\begin{proof}[Proof of Theorem~\ref{thm:main-product-separation}]
Index the roots by $i\in\mathbb{Z}_k$. The prescribed responses are
$y_t^{\rm del}(i,a^t)=\mathbf{1}[t=L]$ and
$y_t^{\rm pay}(i,a^t b^j c)=\mathbf{1}[i+j\equiv0\pmod{k}]$.

For the local upper bound, each $\mathcal U_t^{\rm pay}$ is realized on the $k$ basis states
$s_i$, with $a$ and $c$ acting identically, $b$ cyclically shifting the index,
  and a one-hot final output vector. Each $\mathcal U_t^{\rm del}$ has a one-state realization.
Taking the maximum over these specifications gives a local dimension at most $k$.
The same basis states and independently chosen labelled query effects give the
query-specific static upper bound for the full specification family.

For the local lower bound, fix $\mathcal U_t^{\rm pay}$ and write
$\mu_i:=p_iM_{a^t}$ and $f_j:=M_{b^jc}e$.  If
$z\in\operatorname{supp}(\mu_i)\cap\operatorname{supp}(\mu_{i'})$ for
$i\ne i'$, choose $j\equiv-i\pmod{k}$.  The endpoint equations
$\mu_if_j=1$ and $\mu_{i'}f_j=0$ force both $f_j(z)=1$ and $f_j(z)=0$.
Hence the $k$ supports $\operatorname{supp}(\mu_i)$ are pairwise disjoint, so
every local realization of $\mathcal U_t^{\rm pay}$ has at least $k$ states.  The submatrix
on the roots $i$ and the queries
$a^tc,a^tbc,\ldots,a^tb^{k-1}c$ is a $k\times k$ permutation matrix.  Any
query-specific static factorization $\mathbf R=XE$ with shared roots
therefore has $d\ge\operatorname{rank}(\mathbf R)=k$.

For the shared upper bound, use states $s_{i,t}$ with
$i\in\mathbb{Z}_k$ and $0\le t\le L$, root $i$ initialized at $s_{i,0}$,
and deterministic updates
\[
a:s_{i,t}\mapsto s_{i,\min(t+1,L)},
\qquad
b:s_{i,t}\mapsto s_{i+1,t},
\]
together with
\[
c:s_{i,t}\mapsto s_{i,0}\quad(i\ne0),
\qquad
c:s_{0,t}\mapsto s_{0,L}.
\]
The common final output vector is
$e(s_{i,t})=\mathbf{1}[t=L]$. These transitions realize all prescribed
responses using $k(L+1)$ states.

For the matching lower bound, let $p_{i,t}$ be the state distribution after
$a^t$ from root $i$. Row-stochasticity permits a positive-probability path
$v_{i,0}\to\cdots\to v_{i,L}$ in the transition graph. If
$v_{i,t}=v_{i',t}$ for $i\ne i'$, choose
$j\in\{0,\ldots,k-1\}$ with $j\equiv-i\pmod{k}$. The suffix $b^j c$ then gives
responses one and zero from the same state, a contradiction.

If $t<t'$ and $v_{i,t}=v_{i',t'}$, apply the common suffix
$a^{L-t'}$. From the first occurrence the total exponent is
$t+L-t'<L$, so the response is zero; from the second occurrence the total
exponent is $L$, so the response is one. This is again a contradiction.
Because the selected prefix path has positive probability and the corresponding
response is exactly zero or one, every continuation with positive probability
must have the corresponding final-output value. Thus all
$k(L+1)$ selected path states are distinct, and every exact shared realization for
the specification family
has dimension at least $k(L+1)$.
\end{proof}

The two lower-bound suffixes serve complementary roles. The words $b^j c$
separate payloads uniformly across prefix lengths, while powers of $a$ separate
prefix lengths while preserving the active payload. Their Cartesian interaction
forces the product state set, with each local specification exposing one factor.
Independent output vectors install the delay columns directly on the same
$k$-state root carrier, whereas shared word composition retains the elapsed
prefix length until those columns become reachable.

\section{Response geometry and rank-tight closure}
\label{sec:rank-tight-geometry}

The response table admits a natural convex geometric interpretation. Each latent
state induces a response row over a calibrated query set, and each root response
is a convex combination of these latent profiles. At the minimum state dimension
compatible with the affine dimension, the latent response polytope collapses to a
simplex, transforming the problem of sharing dynamic transitions into the
simultaneous invariance of that simplex. The static counterpart corresponds to
the classical nested-polytope characterization underlying stochastic
factorization and restricted nonnegative matrix factorization
~\citep{mond2003stochastic,gillis2012geometric,chistikov2016restricted}:
an intermediate polytope with minimal vertices is sought between the convex hull
of the root rows and the outer response polytope. Theorem~\ref{thm:rank-tight-closure}
supplements this static containment with the dynamic invariance conditions that
characterize rank-tight shared realization in this calibrated setting.
The additional content is not a new nested-polytope factorization, but the
simultaneous invariance forced by one shared stochastic action family.

Throughout this section, $y(h,u)$ denotes a single unlabelled response table,
obtained either from one specification or by dropping specification labels in a family whose
overlapping declarations agree.  Thus every declared root--word pair has one
target. The invariant-simplex theorem below is stated for this unlabelled
geometric specialization; the labelled/query-specific static relaxation is the
benchmark used in the finite-specification complexity results.

\begin{definition}[Calibrated specification]
\label{def:calibrated}
Let $Y\subseteq\Gamma^*$ be a finite query set containing the empty word
$\varepsilon$, with every pair $(h,v)\in H\times Y$ declared. Write
$r_h(v):=y(h,v)$ for the corresponding root response row. Let
$H_0\subseteq H$ be a set of $k$ roots whose response rows form an affine basis
of
$\mathcal A:=\operatorname{aff}\{r_h:h\in H\}$. The specification is
\emph{calibrated of rank $k$} if the root rows have ordinary linear rank $k$
and affine dimension $k-1$. For each primitive control letter $c\in\Gamma$, all
pairs $(h,cv)$ with $h\in H_0$ and $v\in Y$ are declared; their successor
profiles $r_{h,c}(v):=y(h,cv)$ belong to $\mathcal A$ and uniquely determine
the affine map $F_c:\mathcal A\to\mathcal A$ satisfying
$F_c(r_h)=r_{h,c}$ for $h\in H_0$. For a general control word
$u=a_1\cdots a_t\in\Gamma^*$, let
$F_u:=F_{a_t}\circ\cdots\circ F_{a_1}$, with
$F_\varepsilon:=\operatorname{id}_{\mathcal A}$. Finally, as part of the
definition, the specification is required to be compositionally consistent:
for every declared pair $(h,u)$,
$y(h,u)=(F_u(r_h))(\varepsilon)$; for rational data these finitely many maps
and equalities are checked by rational linear algebra. Write
$\mathcal P:=\mathcal A\cap[0,1]^Y$. Rank-tightness requires attaining a
normalized state dimension equal to the linear-algebraic rank $k$ of the
response table.
\end{definition}

In an exact rank-tight $k$-state shared realization, the $k$ latent response rows are
affinely independent: the root rows span an affine subspace of dimension
$k-1$, and a convex factorization through at most $k$ latent points cannot have
a strictly smaller affine span. In particular, if the latent rows were
affinely dependent, their convex hull would have dimension below $k-1$ and could
not contain the affine basis of root rows. Their convex hull is therefore a
$(k-1)$-simplex $T\subseteq\mathcal P$ with affine hull $\mathcal A$. Shared
stochastic transitions preserve this simplex; conversely, simplex invariance
provides valid barycentric coordinates for the transition rows. The following
theorem formalizes this bidirectional equivalence.

\begin{theorem}[Rank-tight shared-realization closure]
\label{thm:rank-tight-closure}
Let $\mathcal P$ and $\{F_c\}_{c\in\Gamma}$ be given as in
Definition~\ref{def:calibrated}. A compositionally consistent calibrated
rank-$k$ specification admits an exact $k$-state shared realization if and only
if there exists a $(k-1)$-simplex $T$ such that
\begin{equation}
\begin{aligned}
\{r_h:h\in H\}&\subseteq T\subseteq\mathcal P,\\
F_c(T)&\subseteq T\qquad(c\in\Gamma).
\end{aligned}
\label{eq:rank-tight-closure}
\end{equation}

For the root response table on the query set $Y$, a rank-tight static
realization is equivalent to the existence of a $(k-1)$-simplex satisfying
the containment
$\{r_h:h\in H\}\subseteq T\subseteq\mathcal P$.  Shared realization
adds the simultaneous invariance conditions $F_c(T)\subseteq T$.  The static
converse uses the vertex coordinates of $T$ as the independently chosen query
effects.
\end{theorem}

\begin{proof}
$(\Longrightarrow)$ In an exact rank-tight $k$-state shared realization, write
$v_j(w):=\delta_jM_we$ for the response row of latent state $j$ on $Y$.
The $k$ rows $v_1,\ldots,v_k$ are affinely independent, and their convex hull
forms a simplex $T\subseteq\mathcal P$ containing all root response profiles
$\{r_h:h\in H\}$.  Because these root responses contain an affine basis of
$\mathcal A$, the affine hull of $T$ is precisely $\mathcal A$.  Evaluating
the latent rows at the distinguished empty-word coordinate
$\varepsilon\in Y$ yields the final output vector $e\in[0,1]^k$.
The root distributions corresponding to the affine basis rows are themselves
affinely independent, since affine relations among them would map to affine
relations among the root responses; hence they span $\operatorname{aff}(\Delta_k)$.
Because the latent rows are affinely independent as well, the affine map
$\Phi:\operatorname{aff}(\Delta_k)\to\mathcal A$ defined by
$\Phi(\alpha)=\sum_j\alpha_jv_j$ is an affine isomorphism; its restriction maps
$\Delta_k$ onto $T$.  The affine map
\[
 G_c:=\Phi\circ(\alpha\mapsto\alpha M_c)\circ\Phi^{-1}
\]
maps $T$ into itself.  For every $h\in H_0$ and $w\in Y$,
\[
 (G_c(r_h))(w)=p_hM_cM_we=y(h,cw)=r_{h,c}(w).
\]
Thus $G_c(r_h)=F_c(r_h)$ on the affine basis $(r_h)_{h\in H_0}$, so
$G_c=F_c$ on $\mathcal A$ and hence $F_c(T)\subseteq T$.

$(\Longleftarrow)$ Conversely, let $v_1,\ldots,v_k$ be the vertices of a simplex
$T$ satisfying~\eqref{eq:rank-tight-closure}. Because the vertices are
affinely independent, every point in $T$ admits unique barycentric coordinates.
Expressing each root response row $r_h\in T$ in these coordinates defines a
unique initial distribution $p_h\in\Delta_k$. For each primitive letter
$c\in\Gamma$, the invariance condition $F_c(T)\subseteq T$ guarantees that
$F_c(v_j)\in T$ for each vertex $v_j$. Its unique barycentric decomposition
$F_c(v_j)=\sum_{\ell=1}^k M_c(j,\ell)v_\ell$ consists of nonnegative
coefficients summing to one, defining a valid row-stochastic transition matrix
$M_c\in\mathbb R^{k\times k}$.

Define the final output vector $e\in[0,1]^k$ by evaluating each vertex $v_j$ at the
empty-word coordinate $\varepsilon\in Y$, which is valid since
$T\subseteq\mathcal P\subseteq[0,1]^Y$.  Put
$\Phi(\alpha):=\sum_j\alpha_jv_j$.  The vertex definitions give
$\Phi(\alpha M_c)=F_c(\Phi(\alpha))$ and hence
$\Phi(\alpha M_u)=F_u(\Phi(\alpha))$.  For each $v\in Y$, the affine
functionals $x\mapsto(F_v(x))(\varepsilon)$ and $x\mapsto x(v)$ agree on every
$r_h$, $h\in H_0$, because compositional consistency gives
$(F_v(r_h))(\varepsilon)=y(h,v)=r_h(v)$.  Since these rows form an affine
basis of $\mathcal A$, $(F_v(x))(\varepsilon)=x(v)$ for all
$x\in\mathcal A$.  Since
$\alpha e=\Phi(\alpha)(\varepsilon)$, each declared pair satisfies
$p_hM_ue=\Phi(p_hM_u)(\varepsilon)=(F_u(r_h))(\varepsilon)=y(h,u)$.
Thus the constructed tuple reproduces every declared target exactly.

For the static assertion, a width-$k$ factorization
$r_h(w)=p_he_w$ gives latent profiles
$v_j:=(e_w(j))_{w\in Y}\in[0,1]^Y$ whose convex hull contains every $r_h$.
The root rows have affine dimension $k-1$, while there are only $k$ profiles;
hence the profiles are affinely independent and their convex hull is a
$(k-1)$-simplex $T$ with affine hull $\mathcal A$.  Therefore
$T\subseteq\mathcal A\cap[0,1]^Y=\mathcal P$.  Conversely, from such a
simplex take the barycentric coordinates of each $r_h$ as $p_h$ and, for every
$w\in Y$, define the independent query effect by $e_w(j):=v_j(w)$.  Then
$e_w\in[0,1]^k$ and $p_he_w=r_h(w)$ for every $h\in H$ and $w\in Y$.
\end{proof}

For a word $u=a_1\cdots a_t$, states evolve by sequential right multiplication
in the order $a_1,\ldots,a_t$, whereas effects are pulled back in reverse order
starting with $M_{a_t}e$. Consequently, the induced affine action on response
profiles is $F_u=F_{a_t}\circ\cdots\circ F_{a_1}$. This composition order is
used throughout the closure argument. Whereas the static problem uses the
simplex for convex representability, shared realization adds
simultaneous invariance under every primitive generator.

\paragraph{Linear-rank tightness and cellwise saturation.}
Sections~\ref{sec:rank-tight-geometry} and~\ref{sec:lics-finite-partial-specification} instantiate
one closure principle under distinct tightness hypotheses.
Theorem~\ref{thm:rank-tight-closure} uses linear-rank tightness: static
containment supplies one simplex, and shared-transition closure requires simultaneous
invariance.  The reductions use $q$ disjoint endpoint cells; identity
calibration and $K=K_0q$ saturate each with $K_0$ states.  There
$WE=I_{K_0}$ identifies a stochastic basis, and probes make the macros act
coherently between the resulting simplices.

Thus Theorem~\ref{thm:rank-tight-closure} is the global geometric prototype,
while Lemma~\ref{lem:lics-fixed-interface} is the formal, blockwise interface
used by both hardness reductions.  The hard instances are tight in normalized
nonnegative-factorization dimension---$C_{\rm loc}=C_{\rm stat}=K$---without identifying
$K$ with the ordinary rank of the global labelled table.  Above rank tightness,
affine dependence destroys unique barycentric coordinates; a general closure
theorem must retain an intermediate polytope with compatible stochastic lifts.

\section{Complexity of Finite Partial Specifications}
\label{sec:lics-finite-partial-specification}
\label{sec:finite-partial-specification-classification}

We now study explicitly listed rational finite partial specifications, which
give the computational form of the problem defined earlier. Let $N$ denote the
total binary encoding length of the input, including the rational targets, the
alphabet, and the state budget $K$, and let $\ell$ be the length of the longest
declared word. The state budget $K$ is either given in unary or polynomially
bounded in $N$. In the robust variant, the rational tolerance $\epsilon$ has
polynomial encoding length. The exact problem asks whether $J_K(I)=0$, whereas
the robust problem distinguishes $J_K(I)=0$ from $J_K(I)\ge\epsilon$.

\paragraph{Common hardness-output regime.}
For each theorem below, the hardness reduction can be chosen so that every
output instance has constant-bit rational targets and satisfies:
\begin{enumerate}[label=\textup{(H\arabic*)},leftmargin=*]
\item the alphabet is the fixed set
$\{\pl0,\pl1,\#,\pl P,\pl T\}$;
\item every compiled specification contains the full identity-and-boundary constraints
described below and has exact local minimum $K$;
\item the query-specific static relaxation of the whole family has
exact minimum $K$; and
\item overlapping declarations are consistent: if
$(h,u)\in D_\lambda\cap D_{\lambda'}$, then
$y_\lambda(h,u)=y_{\lambda'}(h,u)$.
\end{enumerate}
The reductions certify the equalities in \textup{(H2)}--\textup{(H3)} on every
output, and Lemma~\ref{lem:lics-fixed-interface} converts these
nonnegative-width bounds into exact cell exhaustion at budget $K$. Every YES
output admits an exact $K$-state shared realization whose common final output
vector is Boolean.
Condition \textup{(H3)} is the labelled/query-specific static benchmark of
Section~\ref{sec:controlled-models}; unlabelled static width is a separate
parameter.

\paragraph{Basis calibration as an observable interface.}
The named initial distributions expose the rows of a fixed identity-and-phase interface
with an explicit $K$-state tensor-parser realization.  Its endpoint responses
force every realization of dimension at most $K$ to have exactly $K$ states,
partitioned into $q$ saturated $K_0$-state cells; identity queries fix a
common payload basis in each cell, up to relabeling.  These declarations supply
the lower bounds in \textup{(H2)}--\textup{(H3)}, while explicit per-specification
witnesses and the basis-root labelled factorization supply the matching upper
bounds. The source instance then appears only in finite schedules and rational
routing constraints on the shared partially specified map $V$. Each
specification may choose its own feasible $V$ locally, whereas a shared
 realization for the union table obtained after dropping specification labels
must use one $V$ satisfying all source constraints. Thus the reductions begin
with local and query-specific dimensions fixed at $K$ and place the decision content
in common word-composition compatibility.

\begin{theorem}[Exact completeness for finite partial specifications]
\label{thm:finite-partial-specification-exact}
Under the encoding assumptions above, deciding whether an explicitly listed
rational family of finite partial specifications admits an exact shared realization with at most $K$ states is
$\exists\mathbb{R}$-complete. The same holds when the common final output vector
is required to be Boolean. Both hardness statements hold in the output regime
\textup{(H1)}--\textup{(H4)}.
\end{theorem}

\begin{theorem}[Robust completeness for finite partial specifications]
\label{thm:finite-partial-specification-robust}
Under the same encoding assumptions, the promise problem with input
$(I,K,\epsilon)$,
\begin{equation}
J_K(I)=0
\qquad\text{versus}\qquad
J_K(I)\ge\epsilon,
\label{eq:finite-partial-specification-promise}
\end{equation}
is $\mathsf{PromiseNP}$-complete; so is the variant with $J_K^{\rm B}$ in
place of $J_K$.  There is an absolute $c>0$ such that hardness holds with
$\epsilon\ge N^{-c}$.  Both hardness statements hold in the output regime
\textup{(H1)}--\textup{(H4)}.  On every NO output, the bound
$J_K(I)\ge\epsilon$ holds after minimizing over all $d\le K$ and all effects
$e\in[0,1]^d$, and therefore also establishes the Boolean-final-vector variant.
\end{theorem}

For the compiled robust instances, if $N$ denotes the input encoding length,
the parser has $q=O(N^2)$ sectors, the state budget is $K=K_0q$, and
$L_Q=O(\log N)$.  A PromiseNP certificate consists of
$|H|K+|\Gamma|K^2+K$ dyadic entries, with
$O(\log K+\log\ell+\log(1/\epsilon))$ bits per entry; its total bit length
and all exact verification checks are therefore polynomial in $N$.

Figure~\ref{fig:proof-architecture} summarizes the proof architecture.
Appendix~\ref{sec:finite-partial-specification-appendix} gives the constructions, converse
arguments, and coefficient-level estimates.

\paragraph{Proof architecture.}
Both lower bounds follow the common pipeline in
Figure~\ref{fig:proof-architecture}.
At a high level, the construction (i) encodes source variables in a partially
specified stochastic map, (ii) exposes its basis with identity queries, (iii)
certifies reset sectors with phase probes, (iv) recovers source gates with
payload probes, and (v) transfers exact or robust source constraints.
At the source level, one partially specified stochastic map carries the
continuous or Boolean variables. Identity constraints make its rows observable.
The parser compiles all source gates into five primitive letters,
while boundary probes certify a common reset sector after every complete
macro, and payload probes recover the induced source gates. The resulting
specifications retain exact local and query-specific minimum $K$. The exact and robust branches
share the last three stages and use independent source reductions and
soundness arguments.  Geometrically, calibration exposes simplex vertices
inside each saturated cell, and the compiled declarations force those vertices
to evolve under one shared family of stochastic maps---the cellwise form of
the closure principle in Section~\ref{sec:rank-tight-geometry}.

\begin{figure*}[t]
\centering
\tikzset{
  pa-cardbox/.style={
    draw=black,
    fill=black!1,
    line width=0.55pt,
    rounded corners=0.88mm
  },
  pa-sourcebox/.style={pa-cardbox,fill=black!6},
  pa-compilerbox/.style={
    draw=black,
    fill=black!8,
    line width=0.66pt,
    rounded corners=0.88mm
  },
  pa-invariantbox/.style={
    draw=black,
    fill=black!3,
    line width=0.44pt,
    rounded corners=0.66mm
  },
  pa-classification/.style={
    draw=black,
    fill=black!1,
    line width=0.55pt,
    rounded corners=0.88mm
  },
  pa-flow/.style={
    draw=black!88,
    fill=none,
    line width=0.66pt,
    line cap=rect,
    shorten >=0pt,
    -{Stealth[length=2.05mm,width=1.75mm]}
  },
  pa-divider/.style={draw=black!15,line width=0.44pt},
  pa-title/.style={font=\sffamily\bfseries\fontsize{10.5}{12.5}\selectfont,text=black!88},
  pa-body/.style={font=\sffamily\fontsize{9}{11}\selectfont,text=black!88},
  pa-secondary/.style={font=\sffamily\fontsize{8.25}{10}\selectfont,text=black!60},
  pa-edge/.style={font=\sffamily\fontsize{8.25}{10}\selectfont,text=black!60},
  pa-macro/.style={font=\ttfamily\fontsize{8.25}{10}\selectfont,text=black!88}
}
\resizebox{.96\textwidth}{!}{%
\begin{tikzpicture}[x=0.022cm,y=-0.022cm]
  \path[use as bounding box] (0,42) rectangle (760,560);

  \draw[pa-cardbox] (18,70) rectangle (156,136);
  \node[pa-title,anchor=base] at (87,96) {ETR--INV};
  \node[pa-secondary,anchor=base] at (87,118) {continuous variables};

  \draw[pa-sourcebox] (204,58) rectangle (394,148);
  \node[pa-title,anchor=base] at (299,88) {Identity-constrained};
  \node[pa-title,anchor=base] at (299,107) {ETR source};
  \node[pa-secondary,anchor=base] at (299,130) {exact gate equations};

  \draw[pa-cardbox] (18,206) rectangle (156,272);
  \node[pa-title,anchor=base] at (87,232) {$(3,4)$-SAT};
  \node[pa-secondary,anchor=base] at (87,254) {Boolean variables};

  \draw[pa-sourcebox] (204,194) rectangle (394,284);
  \node[pa-title,anchor=base] at (299,224) {Identity-constrained};
  \node[pa-title,anchor=base] at (299,243) {Boolean source};
  \node[pa-secondary,anchor=base] at (299,266) {robust extraction};

  \draw[pa-flow] (156,103) -- (204,103);
  \draw[pa-flow] (156,239) -- (204,239);
  \node[pa-edge,anchor=base] at (180,94) {exact};
  \node[pa-edge,anchor=base] at (180,230) {gap};

  \draw[pa-compilerbox] (456,92) rectangle (738,252);
  \node[pa-title,anchor=base] at (597,120)
    {Five-letter word parser};
  \draw[pa-divider] (474,135) -- (720,135);

  \node[pa-macro,anchor=base west] at (480,160)
    {$\mathtt{P}\,c_i\,\mathtt{\#}$};
  \node[pa-body,anchor=base west] at (558,160)
    {execute source gate $G_i$};
  \node[pa-macro,anchor=base west] at (480,187)
    {$\mathtt{P}\,d_y\,\mathtt{\#}$};
  \node[pa-body,anchor=base west] at (558,187)
    {read payload coordinate $y$};
  \node[pa-macro,anchor=base west] at (480,214)
    {$\mathtt{T}^{r}$};
  \node[pa-body,anchor=base west] at (558,214)
    {certify sector boundaries};
  \node[pa-secondary,anchor=base] at (597,238)
    {shared parser, calibration and routing};

  \draw[pa-flow] (394,103)
    .. controls (428,103) and (428,142) .. (456,142);
  \draw[pa-flow] (394,239)
    .. controls (428,239) and (428,202) .. (456,202);

  \draw[pa-cardbox] (239,304) rectangle (521,382);
  \node[pa-title,anchor=base] at (380,334) {Finite partial specifications};
  \node[pa-body,anchor=base] at (380,356)
    {one shared transition system};
  \draw[pa-flow] (585,252)
    .. controls (575,266) and (535,271) .. (500,278)
    .. controls (480,282) and (463,290) .. (460,304);

  \draw[pa-invariantbox] (42,410) rectangle (718,458);
  \node[pa-title,anchor=base west] at (62,439) {Preserved:};
  \node[pa-body,anchor=base west] at (166,439)
    {alphabet $\{0,1,\#,\mathsf{P},\mathsf{T}\}$};
  \draw[pa-divider] (320,421) -- (320,447);
  \node[pa-body,anchor=base west] at (342,439) {local optimum $=K$};
  \draw[pa-divider] (490,421) -- (490,447);
  \node[pa-body,anchor=base west] at (512,439) {query-specific $=K$};
  \draw[pa-flow] (380,382) -- (380,410);

  \draw[pa-classification] (42,486) rectangle (368,544);
  \node[pa-title,anchor=base] at (205,510) {Exact branch};
  \node[pa-body,anchor=base] at (205,531)
    {zero defect $\,\cdot\,$ $\exists\mathbb{R}$-complete};

  \draw[pa-classification] (392,486) rectangle (718,544);
  \node[pa-title,anchor=base] at (555,510) {Robust branch};
  \node[pa-body,anchor=base] at (555,531)
    {inverse-polynomial gap $\,\cdot\,$ $\mathsf{PromiseNP}$-complete};
\end{tikzpicture}%
}
\caption{Proof architecture for the finite partial specification complexity classification.
The exact and robust source instances are identity-constrained and compiled
through a common five-letter parser. The construction preserves the fixed
alphabet, keeps every specification's exact local minimum at $K$, and keeps the exact
query-specific static minimum at $K$, yielding
$\exists\mathbb{R}$-completeness at zero defect and
$\mathsf{PromiseNP}$-completeness for an inverse-polynomial promise gap.}
\label{fig:proof-architecture}
\Description{Two source reductions, one exact and one robust, pass through
identity-constrained source gadgets and a shared five-letter word parser to produce finite partial specifications while preserving the alphabet and
the local and query-specific optimum $K$.}
\end{figure*}

\subsection{Membership}
Any $d<K$ realization can be padded to dimension $K$ by adjoining unreachable
absorbing states, extending every root by zeros, every transition by an
identity block, and the final output vector by zeros. It therefore suffices to use
dimension exactly $K$.  Introduce variables for every root distribution,
primitive stochastic matrix, final output vector, and every prefix distribution of
a listed word.  The update equations are quadratic and the response equations
are bilinear, giving a polynomial-size existential real formula; for the
Boolean-final-vector restriction add $e_j(1-e_j)=0$ for every $j$.

For robust membership, take an exact $K$-state witness.  In every root and
transition row, round the first $K-1$ coordinates down to $B$ binary places and
put the remaining mass in the last coordinate. Round the final output vector
coordinatewise down to the same dyadic grid (or retain its exact bits in the Boolean-final-vector
variant).  Stochastic contraction and telescoping show that a word of length at
most $\ell$ changes its response by at most $2K(\ell+2)2^{-B}$.  Taking
\[
 B\ge\left\lceil\log_2\frac{6K(\ell+2)}{\epsilon}\right\rceil
\]
gives a polynomial-size rational certificate of defect at most $\epsilon/3$.
 A verifier checks stochasticity and all labelled responses by exact rational
arithmetic and accepts at threshold $\epsilon/2$; every promised NO instance
has defect at least $\epsilon$ and is rejected.  Every exact rational comparison
has bit length polynomial in $B$, $\ell$, $\log K$, the input-target bit length,
and the bit length of $\epsilon$, so verification runs in polynomial time.
Appendix~\ref{app:finite-partial-specification-membership} gives the explicit bit bound.

\subsection{Calibrated five-letter interface}
By an \emph{uncompiled source system} we mean the source system before parser
compilation: it has a payload state set $Z$ of size $K_0$, root distributions
in $\Delta_{K_0}$, row-stochastic source gates $G_i$, coordinate probes
$e_y$ for $y\in Z$, and finitely many declared root--gate-schedule--probe
 responses.  For a source-gate word $v=A_1\cdots A_t$, write
 $G_v:=A_1\cdots A_t$.  One calibrated root and one coordinate probe for each
 $z\in Z$ produce the identity table
$R(h_z,e_y)=\mathbf 1\{z=y\}$.

\paragraph{Literal parser.}
\label{sec:five-letter-parser}
The parser is a finite control layer placed beside the payload space $Z$.
Intuitively, $Z$ stores the source payload, while $Q$ is a finite parser
register: the letters $0$ and $1$ append a binary code, $\#$ executes the
corresponding source gate or probe and then clears the code, $P$ resets the
register, and $T$ advances the phase used by the boundary signatures. The
product space $Z\times Q$ therefore separates payload memory from control
bookkeeping; no additional physical or probabilistic assumption is introduced.
Let $G_1,\ldots,G_m$ be the source gates and put
\[
\begin{aligned}
 h&:=\max\{1,\lceil\log_2(m+K_0)\rceil\},\\
 Q_{\rm code}&:=\{0,1\}^{\le h}\cup\{\bot\},
 &Q_{\rm ph}&:=\mathbb Z_{2^h}.
\end{aligned}
\]
Choose, for example lexicographically, distinct $h$-bit words $c_i$ for the
gates and $d_y$ for the payload coordinates.  Then
\[
 Q:=Q_{\rm ph}\times Q_{\rm code},\qquad
 q:=|Q|=2^{2h+1},\qquad L_Q:=\log_2q=2h+1.
\]
Thus $q=O((m+K_0)^2)$ and the compiled state budget
$K=K_0q$ is polynomial in the source size; the parser words and all copied
queries have polynomial length and polynomial-bit rational targets.
Enumerate $Q$ as $\rho_0,\ldots,\rho_{q-1}$ and put
$\tau(\rho_j):=\rho_{j+1\bmod q}$.  Compute a cyclic binary de Bruijn word
$\beta_0\cdots\beta_{q-1}$ of order $L_Q$ and put
$\varphi_Q(\rho_j):=\beta_j$.  Then
\begin{equation}
 \chi(b):=(\varphi_Q(b),\varphi_Q(\tau b),\ldots,
                 \varphi_Q(\tau^{L_Q-1}b))
 \label{eq:lics-phase-signature}
\end{equation}
runs bijectively through $\{0,1\}^{L_Q}$.  Write
$\operatorname{reset}(\phi,c):=(\phi,\varepsilon_{\rm code})$.
A cyclic de Bruijn word is obtained by an Euler tour of the binary de Bruijn
graph in time $O(q)$.  Since $q=2^{2h+1}=O((m+K_0)^2)$, the codewords,
$\tau$, and $\varphi_Q$ are uniformly constructible in polynomial time and
have polynomial encoding length.

Given a source witness, the following rules instantiate the five primitive
matrices on $Z\times Q$, with $z_*$ parked; choose $b_0,b_1\in Q$ such that
$\varphi_Q(b_j)=j$ for $j\in\{0,1\}$.
The reduction records their required behavior through the finite labelled
response declarations following \eqref{eq:lics-parser-table}.
\begin{equation}
\begin{aligned}
 \delta_{(z,b)}M_{\pl T}
   &=\delta_{(z,\tau b)},\\
 \delta_{(z,(\phi,c))}M_{\pl P}
   &=\delta_{(z,(\phi,\varepsilon_{\rm code}))},\\
 \delta_{(z,(\phi,c))}M_\sigma
   &=\delta_{(z,(\phi,c\sigma))}
     &&(\sigma\in\{\pl0,\pl1\},\ |c|<h),\\
 \delta_{(z,(\phi,c_i))}M_\#
   &=(\delta_zG_i)\otimes
     \delta_{(\phi,\varepsilon_{\rm code})},\\
 \delta_{(z,(\phi,d_y))}M_\#
   &=\delta_{(z_*,b_{\mathbf1\{z=y\}})}.
\end{aligned}
\label{eq:lics-parser-table}
\end{equation}
A bit read from a complete buffer, or from $\bot$, goes to $\bot$.  On $\#$,
every unused, partial, or malformed code uses a fixed default stochastic row
and resets the code register.  The letter $\pl P$ also resets $\bot$, so every
row is defined. Finally set the sole final output vector to
$e(z,b):=\varphi_Q(b)\in\{0,1\}$.

Writing $\omega_i:=\pl P c_i\#$ and $\zeta_y:=\pl P d_y\#$, direct
multiplication gives, for every payload distribution $p$ and sector $b$,
\begin{equation}
\begin{aligned}
 (p\otimes\delta_b)M_{\omega_i}
   &=pG_i\otimes\delta_{\operatorname{reset}(b)},\\
 (p\otimes\delta_b)M_{\zeta_y}e&=p_y,\\
 (p\otimes\delta_b)M_{\pl T^r}e&=\chi_r(b),\\
 (p\otimes\delta_b)M_{\omega_i\pl T^r}e
   &=\chi_r(\operatorname{reset}(b)).
\end{aligned}
\label{eq:lics-parser-semantics}
\end{equation}
For a specification label $\lambda$, let $H_\lambda$ be its source roots, let
$I_\lambda$ be the source gates named or occurring in that specification, and let
$I_\lambda^{\rm full}\subseteq I_\lambda$ be the completely declared gates.
For $\mathbf i=(i_1,\ldots,i_t)$ put
$G_{\mathbf i}:=G_{i_1}\cdots G_{i_t}$ and
$\Omega_{\mathbf i}:=\omega_{i_1}\cdots\omega_{i_t}$, with empty products
equal to the identity and $\varepsilon$, respectively.  Write
\begin{equation}
 \mathcal S_\lambda:=\{(g,\mathbf i,y;\eta):
 p_gG_{\mathbf i}e_y=\eta\text{ is declared in }\lambda\}.
 \label{eq:finite-source-compilation}
\end{equation}
Its
complete-macro boundary set is
\[
 \mathcal B_\lambda:=\{\varepsilon\}\cup\{\omega_i:i\in I_\lambda\}
 \cup\{\omega_{i_1}\cdots\omega_{i_j}:
 (g,\mathbf i,y;\eta)\in\mathcal S_\lambda,\ 1\le j\le t\}.
\]
For $v\in\mathcal B_\lambda$ and $0\le r<L_Q$, put
$\Theta_b(v,r):=\chi_r(b)$ if $v=\varepsilon$, and
$\Theta_b(v,r):=\chi_r(\operatorname{reset}(b))$ otherwise.  The compiled
label consists of the following indexed ledger.  For a calibrated root $h_z$,
abbreviate $h_{h_z,b}$ as $h_{z,b}$ and write $p_{h_z}=\delta_z$.
\begin{description}[leftmargin=2.8em,labelwidth=2.4em,itemsep=.2ex,
                    topsep=.3ex,parsep=0pt]
\item[\rm(D0)] \emph{Roots.} Introduce $h_{g,b}$ for every
  $g\in H_\lambda$ and $b\in Q$.
\item[\rm(D1)] \emph{Root basis identification.} For $g\in H_\lambda$, $b\in Q$,
  and $y\in Z$, declare $y_\lambda(h_{g,b},\zeta_y)=p_g(y)$.
\item[\rm(D2)] \emph{Boundary signatures.} For $g\in H_\lambda$,
  $b\in Q$, $v\in\mathcal B_\lambda$, and $0\le r<L_Q$, declare
  $y_\lambda(h_{g,b},v\pl T^r)=\Theta_b(v,r)$.
\item[\rm(D3)] \emph{Complete-gate basis identification.} For
  $i\in I_\lambda^{\rm full}$, $z,y\in Z$, and $b\in Q$, declare
  $y_\lambda(h_{z,b},\omega_i\zeta_y)=G_i(z,y)$.
\item[\rm(D4)] \emph{Copied source tests.} For every
  $(g,\mathbf i,y;\eta)\in\mathcal S_\lambda$ and $b\in Q$, declare
  $y_\lambda(h_{g,b},\Omega_{\mathbf i}\zeta_y)=\eta$.
\end{description}
Duplicate declarations are merged.  Thus (D2) identifies and routes the sector
cells, (D1) gives $W_bE_b=I_{K_0}$, (D3) recovers every complete source gate,
and (D4) copies every original source-test equation, with overlaps merged.
Malformed primitive strings add no query; their rows receive the total default
semantics in~\eqref{eq:lics-parser-table}.

\begin{lemma}[Boundary-certified fixed interface]
\label{lem:lics-fixed-interface}
Suppose an uncompiled source system over $Z$ has identity constraints, finitely many fully
declared rational stochastic gates, and at most one stochastic gate declared
only through finitely many source tests.  Assume also that every additional
source root is fully specified as an explicitly declared rational
distribution over $Z$.  Each compiled specification includes the full
root-and-boundary ledger (D0)--(D2) for every calibrated root it uses; in
particular, no root occurring in a copied source test is omitted from its
calibration block.  In polynomial time the source can be compiled to the
five-letter alphabet above with $q$ phase--code sectors and budget $K=K_0q$ so
that:
\begin{enumerate}[label=(\roman*)]
\item the source has a shared realization on its $K_0$ payload states if and
only if the compiled specifications admit an exact shared realization with at
most $K$ states;
\item if every compiled specification contains all calibrated roots
$h_{g,b}$ for $g\in Z$ and $b\in Q$, and every source specification is separately
$K_0$-state feasible, the construction
makes every output specification fully calibrated with exact local minimum $K$, and
the query-specific static relaxation of the full family also has
exact minimum $K$; and
\item the number of queries and the encoding lengths of all listed words and
rational targets are polynomial in the source encoding.
\end{enumerate}

The root-identification and boundary blocks (D1)--(D2) make a specification
\emph{fully calibrated}; (D3)--(D4) carry its gate and source constraints.
\end{lemma}

\begin{proof}
For the forward direction, tensor the source realization with the deterministic
parser.  Each complete macro resets the code coordinate and preserves its
phase coordinate, so the product model has $K_0q$ states and realizes every
boundary and source test.

Conversely, consider any exact realization with $d\le K_0q$.  Put
$f_r=M_{\pl T}^{r}e$ and let
\[
 C_b:=\{x\in[d]:(f_0(x),\ldots,f_{L_Q-1}(x))=\chi(b)\}.
\]
An expectation of $0$ or $1$ against an effect in $[0,1]^d$ forces the
support onto coordinates where that effect has the same endpoint value.
Hence all source roots copied into sector $b$ are supported on $C_b$, and
injectivity of $\chi$ makes the $q$ cells disjoint.  Restricting the $K_0$
calibrated roots and payload-probe effects to $C_b$ gives nonnegative matrices
$W_b,E_b$ with
\[
                         W_bE_b=I_{K_0}.
\]
Thus $|C_b|\ge K_0$ for every $b$.  The saturated budget forces
$|C_b|=K_0$ in every cell and exhausts the carrier, so $W_b$ and $E_b$ are
square nonnegative inverse matrices.  Square nonnegative matrices with
nonnegative inverses are monomial: the zero off-diagonal entries of
$W_bE_b=I$ make distinct row supports of $W_b$ disjoint, and invertibility
makes them nonempty.  Since $W_b$ is row-stochastic, it is a permutation
matrix and $E_b=W_b^{-1}$.  After relabeling, every cell therefore exposes the
same payload basis.  Since $f_0=e$, cell exhaustion also makes the terminal
effect Boolean.

The singleton boundary query for $\omega_i$ now forces every basis state in
$C_b$ to be mapped entirely into $C_{\operatorname{reset}(b)}$.  Consequently
this boundary query induces a row-stochastic payload map between the two
calibrated cells.  Post-macro payload probes recover its rows; for the partial
gate, the copied source tests specify its listed equations.  Restriction to any
synchronized sector recovers one shared source realization.  The local lower
bound follows from the same endpoint-support argument.  The labelled effects
give the static lower bound directly: within any specification label, the payload one-hot columns and
the phase-signature columns assign the $K_0q$ calibrated roots pairwise
distinct binary fingerprints.  Endpoint expectations force their supports to
be pairwise disjoint, so the static carrier also has dimension at least $K$.
Tensoring a separately feasible source model with the parser gives each local
upper bound.  After relabeling its calibrated roots, use those common roots for
the static union, place every additional root at its declared mixture, and use
the pulled-back effect from the corresponding local witness independently for
each labelled response column.  For each label the ledger has
polynomially many declarations, and its words and phase code have polynomial
length.  This proves (iii).  Appendix~\ref{app:fixed-interface} gives the
explicit count, tabulates the construction, and verifies each compiled
declaration.
\end{proof}

For the partial gate $V$ used in both source systems below, the routing and
parking invariant gives zero ideal mass on its exceptional states $s,d$
immediately before every call to $V$.  This isolates its undeclared rows in the
exact simulation and yields the exceptional-mass term in the robust transfer
recurrence.

\subsection{Exact branch}
We reduce from ETR--INV, whose input is a conjunction of atoms over variables
$x_1,\ldots,x_n\in[1/2,2]$ of the forms
$x_i=1$, $x_i+x_j=x_k$, and $x_ix_j=1$; the question is whether such an
assignment exists.  This bounded rational normal form is
$\exists\mathbb{R}$-complete~\citep[Theorem~6]{abrahamsen2018artgallery}.
This bounded form uses rational unit, addition, and inversion atoms only; its
instance size is the number of variables and atoms together with their binary
encodings, so the normalization and the subsequent occurrence splitting are
polynomial-time transformations.
For every unit atom $\alpha:x_i=1$, introduce a fresh variable
$y_\alpha\in[1/2,2]$ and replace it by the two atoms
$x_iy_\alpha=1$ and $x_i=y_\alpha$.  The replacement is equivalent: a unit
solution extends with $y_\alpha=1$, while conversely positivity, equality,
and the product equation give $x_i^2=1$ and hence $x_i=1$.  In particular,
the inversion atom uses two distinct variable symbols.  Next replace every
appearance of every variable, including the fresh $y_\alpha$, by a fresh
occurrence variable, and connect the occurrences of each underlying variable
by a spanning tree of copy atoms.  If the resulting constraint list is empty,
append an always-feasible copy atom on two fresh variables so that the specification
family is nonempty.  The postprocessed instance uses only addition, inversion,
and copy atoms, remains polynomial in size, and is feasible exactly when the
original instance is.
Put $x_i=\frac12+\frac32t_i$.  One partially specified stochastic map encodes
the resulting parameter by
\[
  \delta_{u_i}V=t_i\delta_s+(1-t_i)\delta_d,
  \qquad t_i\in[0,1].
\]

\begin{lemma}[Identity-constrained ETR source]
\label{lem:lics-etr-source}
From an ETR--INV instance one can construct, in polynomial time, an
identity-constrained source with one partial map $V$, fully declared rational
routing maps, and constant-length tests such that a common $K_0$-state source
realization exists if and only if the ETR--INV instance is feasible.  Every
one-constraint source specification is separately $K_0$-state feasible, and every
additional initial-distribution root is a fully specified rational mixture of the
calibrated payload roots.
\end{lemma}

\begin{proof}[Proof sketch; full details appear in Appendix~\ref{app:etr-source}]
The change of variables above turns addition, inversion, and equality of split
occurrences into the three acceptance equations
\begin{align}
 \frac{t_i+t_j+1-t_k}{3}&=\frac29, &
 \frac{t_i+t_j+3t_it_j}{5}&=\frac15, &
 \frac{t_i+1-t_j}{2}&=\frac12 .
 \label{eq:lics-etr-tests}
\end{align}
 Let $\mathcal C$ be the postprocessed constraint set.  For each
 $c\in\mathcal C$, let $\mathcal B_c=(\beta_1,\ldots,\beta_m)$ be its ordered
 branch occurrences; repeated branch types are distinct occurrences.  Rational
 initial-distribution roots use the following branches and weights:
\begin{equation}
\begin{aligned}
 x_i+x_j=x_k:\quad&\tfrac13(\mathsf L_i^++\mathsf L_j^++\mathsf L_k^-)
       &&\text{target }\tfrac29,\\
 x_ix_j=1:\quad&\tfrac15\mathsf L_i^++\tfrac15\mathsf L_j^+
                  +\tfrac35\mathsf P_{ij}
       &&\text{target }\tfrac15,\\
 x_i=x_j:\quad&\tfrac12(\mathsf L_i^++\mathsf L_j^-)
       &&\text{target }\tfrac12.
\end{aligned}
\label{eq:lics-etr-gadgets}
\end{equation}

 Put $U:=\{u_i:i\text{ is a split occurrence}\}$ and
 \[
  Z=U\mathbin{\dot\cup}\{s,d,a,r\}
    \mathbin{\dot\cup}
    \{w_\beta,a_\beta^\circ,r_\beta^\circ:
                         c\in\mathcal C,\ \beta\in\mathcal B_c\}.
 \]
 Here $s,d$ are transient Bernoulli states, $a,r$ are the live accept/reject
 states, $w_\beta$ is waiting, and $a_\beta^\circ,r_\beta^\circ$ are parked.
 The declared probes $\delta_{u_i}Ve_y=0$ for every $y\notin\{s,d\}$,
 together with row stochasticity, give the only partial gate $V$ the row
 $t_i\delta_s+(1-t_i)\delta_d$ on $u_i$; it fixes $a,r$ and every waiting and parked
 state, while its $s,d$ rows are undeclared.  Every other gate is deterministic
 and fully specified. The following are all of its nonidentity rows; every
 omitted row is an identity row:
 \begin{equation}
 \begin{array}{@{}c|l@{}}
 \text{gate}&\text{nonidentity rows}\\ \hline
 \mathsf A_\beta&w_\beta\mapsto u_{i(\beta)}\\
 \mathsf R_j&s\mapsto u_j,\quad d\mapsto r\\
 \mathsf F^+&s\mapsto a,\quad d\mapsto r\\
 \mathsf F^-&s\mapsto r,\quad d\mapsto a\\
 \mathsf\Pi_\beta&a\mapsto a_\beta^\circ,\quad
                          r\mapsto r_\beta^\circ\\
 \mathsf C_c&a_\beta^\circ\mapsto a,\quad
              r_\beta^\circ\mapsto r\quad(\beta\in\mathcal B_c).
 \end{array}
 \label{eq:lics-etr-routing-table}
 \end{equation}
 Here $i(\mathsf L_i^\pm)=i(\mathsf P_{ij})=i$.  The three complete branch
 words are
 \begin{equation}
 \begin{aligned}
  \eta_\beta&=\mathsf A_\beta V\mathsf F^+\mathsf\Pi_\beta
      &&(\mathsf L_i^+),\\
  \eta_\beta&=\mathsf A_\beta V\mathsf F^-\mathsf\Pi_\beta
      &&(\mathsf L_i^-),\\
  \eta_\beta&=\mathsf A_\beta V\mathsf R_jV
                    \mathsf F^+\mathsf\Pi_\beta
      &&(\mathsf P_{ij}),
 \end{aligned}
 \label{eq:lics-etr-branch-words}
 \end{equation}
 and the complete test for $c$ is
 \begin{equation}
  p_{h_c}=\sum_{\beta\in\mathcal B_c}\alpha_\beta\delta_{w_\beta},
  \qquad
  v_c=\eta_{\beta_1}\cdots\eta_{\beta_m}\mathsf C_c,
 \label{eq:lics-etr-complete-test}
 \end{equation}
 where the weights $\alpha_\beta$ are those in
 \eqref{eq:lics-etr-gadgets}.
 For each $c$, its source specification consists of the common identity constraints and
  basis-identification queries together with the single test $p_{h_c}G_{v_c}e_a=r_c$; the shared
 family is the labelled union of these specifications.

 Induction over $\beta_1,\ldots,\beta_m$ gives the uniform routing invariant:
 immediately before every occurrence of $V$ in $v_c$, the current distribution
 has zero mass on $\{s,d\}$; its active component is on the relevant $u_i$ (or
 on $u_j$ before the second $V$ of a product branch), and all remaining mass is
 on states fixed by $V$.  Indeed, $\mathsf F^\pm$ removes $s,d$-mass in a
 linear branch, $\mathsf R_j$ removes it before the second $V$ in a product
 branch, and $\mathsf F^+\mathsf\Pi_\beta$ then evaluates and parks that branch.
 Thus the undeclared $s,d$ rows are never used; the three branch acceptance
 probabilities are $t_i$, $1-t_i$, and $t_it_j$, and $\mathsf C_c$ yields
 exactly the convex combinations in \eqref{eq:lics-etr-gadgets}.  Hence $V$ is
 the unique partial gate, all other gates are fully declared rational gates,
  and every $p_{h_c}$ is a fully specified rational mixture, as required by
 Lemma~\ref{lem:lics-fixed-interface}.

One common $V$ supplies one common tuple $t$, so satisfying every specification is
equivalent to satisfying the ETR--INV conjunction.  Conversely, any satisfying
tuple fills the free rows of $V$ and gives the source realization.  Each single
addition, inversion, or copy constraint has a solution in $[1/2,2]$, so its
specification is feasible independently. The initial-distribution weights are declared by
  complete payload basis identification, proving the final assertion.
\end{proof}

Applying Lemma~\ref{lem:lics-fixed-interface} to this source proves exact
hardness and preserves exact local and query-specific optimum $K$.

\subsection{Robust branch}
We use $(3,4)$-SAT: every clause has three literals and every variable occurs
in at most four clauses; this restriction remains NP-complete
~\citep{tovey1984simplified}.  Delete isolated variables and, if the formula
has no clauses, conjoin one clause $z_1\vee z_2\vee z_3$ on fresh variables;
this preserves satisfiability and the occurrence bound.  Canonically rename
the remaining variables consecutively.  If the resulting formula has $p$
variables, indexed by $[p]$, and $q_F$ clauses, put
\[
                 N_{\rm bool}:=\max\{2,p+q_F\}.
\]
Throughout the Boolean reduction, $N_{\rm bool}$ denotes this combinatorial
size.  The standard explicit binary encoding length lies between
$\Omega(N_{\rm bool})$ and $O(N_{\rm bool}\log N_{\rm bool})$, and is therefore
polynomially equivalent to $N_{\rm bool}$.
The same interface carries parameters $b_i\in[0,1]$.  A Booleanity test has
response $b_i(1-b_i)$, and a clause test
has failure response $\prod_{j=1}^3(1-\ell_j)$.  Approximate phase and identity
calibration first recover a basis up to inverse-polynomial error.  All tests
start at calibrated variable roots.  The following quantitative form records
the extraction needed for a gap reduction.

\begin{lemma}[Identity-constrained Boolean source]
\label{lem:lics-boolean-source}
A bounded-occurrence $3$-SAT formula can be mapped in polynomial time to an
identity-constrained source with one partial map, fully declared rational
routing maps, and at most $g_0:=6$ gates per test.  It has a common
exact $K_0$-state realization if and only if the formula is satisfiable; every
single Booleanity or clause specification is separately $K_0$-state feasible.
\end{lemma}

\begin{proof}[Proof sketch; full details appear in Appendix~\ref{app:robust-boolean}]
Use $Z=\{u_1,\ldots,u_p,s,d,a,r\}$.  The partial map has
$\delta_{u_i}V=b_i\delta_s+(1-b_i)\delta_d$; zero probes outside $s,d$ and
row stochasticity enforce this form, while $V$ fixes $a,r$.  Its rows on
$s,d$ are undeclared.  Use fully declared deterministic gates whose only
nonidentity rows are
\begin{equation}
\begin{array}{@{}c|cc@{}}
 &s&d\\ \hline
 R_j^s&u_j&r\\
 R_j^d&r&u_j\\
 F^s&a&r\\
 F^d&r&a
\end{array}
\label{eq:lics-boolean-routes}
\end{equation}
A Booleanity word is $VR_i^sVF^d$, so its accepting mass from $u_i$ is
exactly $b_i(1-b_i)$ and the same partially specified map $V$ is used twice.

For a clause $C=(\ell_1\vee\ell_2\vee\ell_3)$, start at the variable state
of $\ell_1$.  After a positive literal $x_i$, use $R_j^d$ to send its failure
branch $d$ to the next variable $u_j$; after a negative literal $\neg x_i$,
use $R_j^s$ to send its failure branch $s$ to $u_j$.  The complementary mass
goes to $r$ and never returns.  At the final literal use $F^d$ for a positive
literal and $F^s$ for a negative one.  Thus,
if $\nu_j$ is the surviving failure mass,
\begin{equation}
 \nu_0=1,\qquad \nu_j=\nu_{j-1}(1-\ell_j),\qquad
 \operatorname{resp}(C)=\nu_3=\prod_{j=1}^3(1-\ell_j).
\label{eq:lics-clause-schedule}
\end{equation}
Both Booleanity and clause targets are zero.

Immediately before each call to $V$, the active distribution is supported on
a variable state and $r$, which $V$ fixes; hence the declared rows determine
every test, and identity rows complete the remaining transitions.
Booleanity targets enforce $b_i\in\{0,1\}$, clause targets enforce a true
literal, and a satisfying assignment realizes all targets. Each specification is
separately feasible by choosing an endpoint or one true literal.  Booleanity
and clause tests use four and six gates, respectively.
\end{proof}

\begin{lemma}[Robust Boolean extraction]
\label{lem:lics-robust-extraction}
For the compiled source of Lemma~\ref{lem:lics-boolean-source}, set
\[
 c_0:=512,\qquad
 \kappa:=c_0(1+K_0)^2(1+L_Q),\qquad
 C_0:=11(\kappa+10)^{g_0+1},
\]
where $g_0=6$ is the maximum number of source gates in a test.  If a candidate
with $d\le K$ has maximum declared response defect
$\delta\le\min\{1/\kappa,1/(26C_0)\}$, then values
$\widehat b_i\in[0,1]$ can be recovered so that every copied Booleanity or
clause schedule $w$ satisfies
\begin{equation}
 \bigl|\operatorname{resp}(w)-P_w(\widehat b)\bigr|\le C_0\delta,
 \label{eq:lics-word-transfer}
\end{equation}
where $P_w$ is the ideal source response.  Rounding each $\widehat b_i$ to the
nearer Boolean endpoint yields a satisfying assignment.
\end{lemma}

\begin{proof}
Appendix~\ref{app:robust-boolean} records the coefficient-level expansions;
the complete implication is given here in three steps.
\emph{Step 1: cell saturation and approximate basis recovery.}
For each latent state $x$, put $f_r:=M_{\pl T}^{r}e$ and round $f_r(x)$ to
$0$ below $1/2$ and to $1$ at or above $1/2$.  Since $q=2^{L_Q}$ and
$\chi:Q\to\{0,1\}^{L_Q}$ is bijective by
\eqref{eq:lics-phase-signature}, every rounded signature equals $\chi(b)$ for
a unique $b\in Q$.  The resulting $q$ rounded cells are therefore a disjoint
partition of the entire latent carrier.  The endpoint errors leave at most
$2L_Q\delta$ mass of each calibrated root outside its prescribed rounded cell.
If $o\le2L_Q\delta$ is the removed mass in one row, restriction and
conditioning change any $\delta$-accurate identity entry by at most
$(\delta+o)/(1-o)\le2\delta+4o$; here $o<1/2$ follows from
$\delta\le1/\kappa$.  Hence the root and payload-probe matrices in a cell
satisfy
\[
 \|W_bE_b-I_{K_0}\|_{\max}
 \le e_0:=(2+8L_Q)\delta.
\]
The hypothesis $\delta\le1/\kappa$ implies $K_0e_0<1$.
If a cell had fewer than $K_0$ states, then $W_bE_b$ would have
rank below $K_0$, while
$\sigma_{\min}(W_bE_b)\ge1-K_0e_0>0$, a contradiction.  Thus all $q$ cells
have at least $K_0$ states; the budget $d\le K_0q$ forces every cell to contain
exactly $K_0$ states and the cells to exhaust the carrier.  Hence every realization below the calibration threshold
$1/[K_0(2+8L_Q)]$ has $d=K$ and the saturated cell decomposition.

We next record the quantitative rigidity used below.  Since $\delta\le1/\kappa$,
\[
 e_0\le\frac{1}{64(K_0+1)^2}\le\frac{1}{4(K_0+1)}.
\]
For a rounded cell $C_b$, put
$G_i:=\{\xi\in C_b:(E_b)_{\xi i}>1/2\}$ and
$G_i^\circ:=G_i\setminus\bigcup_{j\ne i}G_j$.
The diagonal entries of $W_bE_b$ put mass at least $1-2e_0$ on the threshold
set for row $i$, while every overlap with a different threshold set carries at
most $2e_0$ per state by an off-diagonal entry.  Hence the sets $G_i^\circ$ are
nonempty and disjoint inside the $K_0$-element cell.  They are
singletons and define a permutation $\pi$ satisfying, for $i\ne j$,
\begin{equation}
\begin{aligned}
 (W_b)_{i,\pi(i)}&\ge1-2(K_0+1)e_0,\\
 (E_b)_{\pi(i),i}&\ge1-(2K_0+3)e_0,&
 (E_b)_{\pi(i),j}&\le2e_0.
\end{aligned}
\label{eq:lics-approx-basis}
\end{equation}
The last two estimates follow by substitution in the diagonal and
off-diagonal entries of $W_bE_b$; Appendix~\ref{app:robust-boolean} records the
coefficient audit.

\emph{Step 2: one-gate transfer.}
Fix a synchronized sector $b_\phi=(\phi,\varepsilon_{\rm code})$ and write
$C$ for its rounded cell.  Identifying the calibrated row labels with the
payload set $Z$, the preceding permutation is a bijection $\pi:Z\to C$.  Put
\[
 e_y^{\rm pull}:=M_{\zeta_y}e,\qquad
 \operatorname{dec}(x)_y:=xe_y^{\rm pull},\qquad
 \theta(x):=x(C^c).
\]
Set $n:=K_0$, $L:=L_Q$, and $A_0:=(n+1)^2(1+L)$, so that
$\kappa=512A_0$.  Phase leakage and
\eqref{eq:lics-approx-basis} give, for every calibrated root $p_z$,
\begin{equation}
 \|p_z-\delta_{\pi(z)}\|_1+\theta(p_z)
 +\sum_{y\in Z}|e_y^{\rm pull}(\pi(z))-\mathbf1\{y=z\}|
 \le\kappa\delta.
\label{eq:lics-basis-ledger}
\end{equation}
Writing $x_\pi:=(x_{\pi(z)})_{z\in Z}$, the same decomposition gives
\begin{equation}
 \|\operatorname{dec}(x)-x_\pi\|_1
 \le\kappa\delta+(n+1)\theta(x).
\label{eq:lics-decoded-coordinate-error}
\end{equation}
For a completely declared source gate $G$, post-macro basis identification,
\eqref{eq:lics-basis-ledger}, and stochastic $\ell_1$-contraction yield
\begin{equation}
\begin{aligned}
 \|\operatorname{dec}(xM_G)-\operatorname{dec}(x)G\|_1
   &\le73A_0\delta+3A_0\theta(x),\\
 \theta(xM_G)&\le\theta(x)+50A_0\delta.
\end{aligned}
\label{eq:lics-complete-transfer-ledger}
\end{equation}

For the partially declared map put
$q_i:=p_{u_i}M_V$ and $\widehat b_i:=\operatorname{dec}(q_i)_s$.
The boundary signature, the declared zero probes outside $\{s,d\}$, and
normalization define the source row
$\delta_{u_i}\widehat V=\widehat b_i\delta_s+
(1-\widehat b_i)\delta_d$.  Together with the completely declared fixed rows,
this defines a stochastic map $\widehat V$; its $s,d$ rows may be arbitrary.
The boundary and zero probes bound the discrepancy from this row, while
\eqref{eq:lics-basis-ledger} replaces calibrated roots by basis states.
Collecting the two undeclared rows in
$\mu(x):=x_{\pi(s)}+x_{\pi(d)}$ gives
\begin{equation}
\begin{aligned}
 \|\operatorname{dec}(xM_V)-\operatorname{dec}(x)\widehat V\|_1
   &\le102A_0\delta+(2n+1)\theta(x)\\[-0.2em]
   &\qquad +(n+1)\mu(x),\\
 \theta(xM_V)&\le\theta(x)+50A_0\delta+\mu(x).
\end{aligned}
\label{eq:lics-partial-transfer-ledger}
\end{equation}
The four inequalities
\eqref{eq:lics-complete-transfer-ledger}--
\eqref{eq:lics-partial-transfer-ledger} are the one-gate simulation ledger.
Appendix~\ref{app:robust-boolean} derives every coefficient; the fixed choice
$\kappa=512A_0$ dominates them all.  Finally put
$\beta:=(2n+3)e_0<1/2$ for the exceptional-mass estimate.

\emph{Step 3: schedule invariant and Boolean rounding.}
The proof now follows the declaration ledger: (D2) bounds cell leakage and
certifies reset boundaries, (D1) supplies the decoded payload basis, (D3)
gives the complete-gate lines above, and (D4) defines $\widehat V$ and the
partial-gate lines.  Fix a declared schedule $H_1\cdots H_g$, $g\le g_0$.
Let $x_r:=x_{r-1}M_{H_r}$ and
$v_r:=v_{r-1}\overline H_r$, where $\overline H_r$ is the declared complete
source gate or $\widehat V$, respectively.  Define
\[
 E_r:=\|\operatorname{dec}(x_r)-v_r\|_1,\qquad
 \theta_r:=\theta(x_r),\qquad S_r:=E_r+\theta_r.
\]
These are decoding error, cell leakage, and their sum.  If $H_r=V$, the
routing-and-parking invariant gives $v_{r-1}(s)=v_{r-1}(d)=0$; diagonal basis
rigidity makes the exceptional mass part of the existing decoding error:
\[
 (1-\beta)\mu(x_{r-1})
 \le\operatorname{dec}(x_{r-1})_s+
      \operatorname{dec}(x_{r-1})_d
 \le E_{r-1}.
\]
Since $\beta<1/2$, $\mu(x_{r-1})\le2E_{r-1}\le2S_{r-1}$.  Applying the
appropriate gate-ledger lines and stochastic $\ell_1$-contraction gives the
single prefix recurrence
\begin{equation}
 S_r\le(\kappa+10)S_{r-1}+10\kappa\delta.
\label{eq:lics-simulation-invariant}
\end{equation}
At the calibrated starting root,
$S_0\le(n+2L)\delta\le\kappa\delta$.  Writing $A:=\kappa+10$ and iterating gives
\[
 S_r\le A^r\kappa\delta+
 10\kappa\delta\sum_{j=0}^{r-1}A^j
 \le11A^{r+1}\delta\le C_0\delta
 \qquad(r\le g_0),
\]
the promised uniform simulation invariant.  The final response is one decoded
coordinate, proving \eqref{eq:lics-word-transfer}.  For a Booleanity word, its zero
target gives $\widehat b_i(1-\widehat b_i)\le(C_0+1)\delta\le2C_0\delta$.
Since $\min\{b,1-b\}\le2b(1-b)$ on $[0,1]$, each $\widehat b_i$ lies within
$4C_0\delta$ of a Boolean endpoint.  If the resulting
rounded assignment falsified a clause, its ideal failure response would be at
least $1-12C_0\delta$, and gate transfer would leave the candidate response at
least $1-13C_0\delta\ge1/2$.  Its declared zero target bounds the candidate
response by $\delta<1/2$, a contradiction.
\end{proof}

It follows that an unsatisfiable formula has defect at least
\begin{equation}
 \Delta_{\rm bool}:=
 \min\!\left\{
  \frac1{K_0(2+8L_Q)},\frac1\kappa,\frac1{26C_0}
 \right\}.
\label{eq:lics-gap-threshold}
\end{equation}
Indeed, suppose that $J_K(I)<\Delta_{\rm bool}$.  By the compactness
observation following the definition of $J_K(I)$, the minimum is attained by
some realization of dimension $d\le K$; write
$\delta=J_K(I)<\Delta_{\rm bool}$ for its defect.  The first threshold forces
the saturated sector allocation with $d=K$, while the other two meet the
basis-recovery and Boolean-extraction hypotheses.
Lemma~\ref{lem:lics-robust-extraction} would then produce a satisfying
assignment, a contradiction.  Hence $J_K(I)\ge\Delta_{\rm bool}$.  Here
$K_0=O(N_{\rm bool})$, the parser has
$q=O(N_{\rm bool}^2)$ sectors and $L_Q=O(\log N_{\rm bool})$, while $g_0=6$.
Hence both $\kappa=512(1+K_0)^2(1+L_Q)$ and
$C_0=11(\kappa+10)^{g_0+1}$ are polynomially bounded.  The output contains
one constant-size source test for every variable and clause, and the compiler
expansion is polynomial.  Hence
$N_{\rm bool}\le N_{\rm out}\le N_{\rm bool}^{a}$ for an absolute
$a$.  Consequently, for an absolute $c>0$ the rational threshold
$\epsilon:=\Delta_{\rm bool}$ has polynomial encoding length and satisfies
$\epsilon\ge N_{\rm out}^{-c}$.

Continuously varying gate parameters establish $\exists\mathbb{R}$-hardness,
whereas Boolean isolation gives
the metric gap required for $\mathsf{PromiseNP}$-hardness.

\subsection{Completion of the classification}
Section~5.1 proves exact membership and, by normalization-preserving rounding,
$\mathsf{PromiseNP}$ membership for unrestricted and Boolean effects;
Lemmas~\ref{lem:lics-etr-source} and~\ref{lem:lics-fixed-interface} give exact hardness.

For any $\Pi\in\mathsf{PromiseNP}$, Cook--Levin, standard $3$-CNF conversion,
and Tovey's bounded-occurrence transformation produce a polynomial-size
$(3,4)$-SAT instance satisfiable exactly on the YES side
~\citep{cook1971complexity,tovey1984simplified}.  Lemmas~\ref{lem:lics-boolean-source}
and~\ref{lem:lics-robust-extraction} and
\eqref{eq:lics-gap-threshold} map its YES instances to defect zero and its NO
instances to defect at least $N_{\rm out}^{-c}$, proving robust hardness.

Each source specification has its own satisfying parameter choice, so
Lemma~\ref{lem:lics-fixed-interface}(ii) gives local and query-specific optima
$K$. The five-letter compiler yields exact $K$-state Boolean YES witnesses; on
NO instances, exact infeasibility and the robust gap hold for every
$d\le K$ and $e\in[0,1]^d$. Calibration and routing keep all overlap targets
consistent.

\section{Discussion}
\label{sec:lics-discussion}

Shared realization asks whether finite trace tests admit a single stochastic
implementation within a state budget. The separation result shows that
per-specification optimal realizations need not compose into an equally small
shared model, even when the collection admits a query-specific static
factorization of the same width. The extra states are the price of enforcing
word-composition compatibility. In the rank-tight regime, static feasibility
seeks a simplex inside the response polytope, whereas shared feasibility
additionally demands forward invariance under every letter update; our
reductions enforce this distinction cellwise through an exposed basis.

The complexity classifications sharpen this separation: exact bounded-state
shared realization is $\exists\mathbb{R}$-complete, while distinguishing zero
defect from inverse-polynomial defect is $\mathsf{PromiseNP}$-complete. Both
hardness results persist over a fixed five-letter alphabet, even when the
local and query-specific static optima already equal $K$. The results remain
valid for Boolean final vectors, and the robust gap persists over the full
cube $[0,1]^d$. Thus the hardness is not tied to a particular choice of
terminal outputs.

The formulation allows several named initial distributions supplied by a
resettable test harness; these are not part of the state space but serve as
basis-synchronization probes. Whether comparable separation and hardness
results hold with a single initial distribution remains open. So does the
unlabelled static benchmark, which collapses per-specification outputs into
one word-level vector and thereby discards test-suite provenance.

Future directions include weaker calibration conditions, structured
transition families, succinct specification formalisms, and tractable
parameter regimes.

\begin{acks}
The author would like to thank Fei Yan and Pan Zhang for their mentorship and guidance. Deep gratitude is also extended to the author's wife, Ling Wang, for her unwavering support and encouragement.
\end{acks}

\bibliographystyle{ACM-Reference-Format-unsrt}
\bibliography{reference}

\appendix
\clearpage

\section{Proofs for finite partial specifications}
\label{sec:finite-partial-specification-appendix}
\label{sec:finite-partial-specification}

This appendix establishes the ingredients of
Theorems~\ref{thm:finite-partial-specification-exact} and~\ref{thm:finite-partial-specification-robust}: exact
and robust membership, the boundary-certified five-letter interface, the
ETR--INV source reduction, and robust Boolean extraction.

\subsection{Membership}
\label{app:finite-partial-specification-membership}

\begin{lemma}[Exact membership for finite partial specifications]
\label{lem:finite-partial-specification-etr-membership}
The exact problem in Theorem~\ref{thm:finite-partial-specification-exact} belongs to
\(\exists\mathbb{R}\).
\end{lemma}
\begin{proof}
Pad any $d<K$ model by unreachable absorbing states, zero root coordinates,
and zero terminal coordinates; this preserves every response and Booleanity of
the effect.  It is therefore enough to use dimension exactly $K$.
Let \(\mathcal V\) be the union of the prefix closures of all listed words.
Introduce a row probability vector \(p_h\) for every root, a row-stochastic
\(K\times K\) matrix \(M_a\) for every letter, an effect
\(e\in[0,1]^K\), and a row vector \(x_{h,v}\) for every
\(h\) and \(v\in\mathcal V\).  Impose
\[
\begin{aligned}
 x_{h,\varepsilon}&=p_h,\\
 x_{h,va}&=x_{h,v}M_a
   &&(v,va\in\mathcal V,\ a\in\Gamma),\\
 x_{h,u}e&=y_I(\lambda,h,u)
 &&((\lambda,h,u)\in D^{\rm tag}(I)).
\end{aligned}
\]
together with nonnegativity, the effect bounds, and all row-sum equations.
Update and response equations have degree at most two; the remaining
constraints are linear.  Their number is
polynomial in the total listed word length.  For the Boolean-final-vector variant,
add $e_j(1-e_j)=0$ for every $j$.  This is the standard
\(\exists\mathbb{R}\) model
\citep{canny1988pspace,renegar1992real,schaefer2026compendium}.
\end{proof}

\begin{lemma}[Normalization-preserving rational certificate]
\label{lem:finite-partial-specification-np-membership}
Under the encoding assumptions of Theorem~\ref{thm:finite-partial-specification-robust}, the
zero-versus-\(\epsilon\) promise belongs to \(\mathsf{PromiseNP}\), both with
an unrestricted effect and with a Boolean effect.
\end{lemma}
\begin{proof}
A zero-error instance has an exact minimizer by compactness and may be padded
to dimension $K$.  In every root distribution and every row of every
transition matrix, round the first \(K-1\) entries down to multiples of
\(2^{-B_{\rm bit}}\) and put the remaining mass in the last entry.  The
\(\ell_1\) error of each rounded row is at most
\(2K2^{-B_{\rm bit}}\).  Round every terminal-effect coordinate down to the
same dyadic grid; in the Boolean-final-vector variant retain its exact bits.
Row-stochastic matrices contract
total variation, so telescoping a word \(u\) gives response error at most
\[
                       2(|u|+2)K2^{-B_{\rm bit}}.
\]
Taking
\[
 B_{\rm bit}\ge\left\lceil\log_2\frac{6K(\ell+2)}{\epsilon}\right\rceil
\]
produces a polynomial-bit rational model of error at most
\(\epsilon/3\).  Indeed, positivity and binary encoding of $\epsilon$ make
$B_{\rm bit}=O(\log K+\log\ell+\log(1/\epsilon))$ polynomial in the input.
The certificate has $|H|K+|\Gamma|K^2+K$ dyadic entries.  Let
$L_{\rm tar}\ge1$ be the maximum binary bit length of a listed rational target
(with value $1$ if no target is listed), and let $L_\epsilon$ be the binary bit
length of $\epsilon$.  Along a word of length at most $\ell$, exact matrix
multiplication produces response numerators and denominators of
$O((\ell+2)(B_{\rm bit}+\log K))$ bits.  Cross-multiplication against a target
and the threshold $\epsilon/2$ therefore uses
$O((\ell+2)(B_{\rm bit}+\log K)+L_{\rm tar}+L_\epsilon)$ bits.  Since both
$L_{\rm tar}$ and $L_\epsilon$ are included in the input length, all checks are
polynomial-time exact rational computations.  A verifier checks stochasticity
and all labelled responses and accepts at threshold \(\epsilon/2\).
Every instance with \(J_K\ge\epsilon\) is rejected; checking the effect bits
gives the same conclusion for $J_K^{\rm B}$.
\end{proof}

\subsection{Identity calibration through the five-letter parser}
\label{app:fixed-interface}

The following interface lemma applies to an identity-constrained finite
source.  Its exact preservation and approximate basis estimates are the
common ingredients of the two reductions.  Write
\[
        \omega_i:=\pl{P}c_i\#
\]
for the compiled macro associated with source map $G_i$.  A legal compound
control is a finite sequence of complete macros.  If
\(\mathcal W_{\mathcal U}\) is the finite set of source schedules occurring
in a specification \(\mathcal U\), define its complete-macro boundary closure by
\begin{equation}
\begin{split}
 \operatorname{BPref}(\mathcal U):={}&\{\varepsilon\}
 \cup\{\omega_i:i\text{ is named or used in }\mathcal U\}\\
 &\cup
 \{\omega_{i_1}\cdots\omega_{i_j}:
       \omega_{i_1}\cdots\omega_{i_t}\in\mathcal W_{\mathcal U},\
       1\le j\le t\}.
\end{split}
\label{eq:boundary-prefix-closure}
\end{equation}
Thus every macro occurring in a noninitial slot is also tested as a
one-macro schedule.

The compilation replaces each source root \(g\) by roots \(h_{g,b}\),
\(b\in Q\), and copies every declared source response into every sector.
For \(b=(\phi,z)\in Q_{\rm ph}\times Q_{\rm code}\), put
\[
 \operatorname{reset}(b):=(\phi,\varepsilon_{\rm code}),
 \qquad
 \chi_r(b):=\varphi_Q(\tau^r b)
 \quad(0\le r<L_Q).
\]
Call a specification \emph{fully calibrated} if it contains the payload identity
calibration and is boundary-complete in the following sense: for every
source root \(g\) used by the specification, every \(b\in Q\), every
\(v\in\operatorname{BPref}(\mathcal U)\), and every \(0\le r<L_Q\), it
contains the declared query
\begin{equation}
 y_\lambda(h_{g,b},v\pl{T}^r)=
 \begin{cases}
   \chi_r(b),&v=\varepsilon,\\
   \chi_r(\operatorname{reset}(b)),&v\ne\varepsilon.
 \end{cases}
\label{eq:boundary-phase-targets}
\end{equation}
These reset-aware \(0/1\) targets are independent of every payload parameter
and cover each partially specified source map.  Fully specified source maps
carry complete post-macro payload basis-identification queries; a partially specified map carries
its declared payload entries and rational source-test responses.

\begin{lemma}[Identity-constrained fixed-interface preservation]
\label{lem:identity-fixed-interface}
Let \(Z\) be a finite payload set, \(K_0:=|Z|\), with one calibrated source
root \(h_z\) and one coordinate-probe output vector \(e_y\) for every
\(z,y\in Z\).  The source identity constraints are
\begin{equation}
 R(h_z,e_y)=\mathbf1\{z=y\};
 \qquad
 y_\lambda(h_{z,b},\pl{P}d_y\#)=\mathbf1\{z=y\}
 \quad(b\in Q)
\label{eq:identity-calibration}
\end{equation}
Here the vectors $e_y$ are only coordinate probes in the source-level
description.  After compilation, each probe is realized by a query word
ending in the same shared final output vector; the construction therefore
introduces no additional terminal effect.  Besides finitely many completely
declared rational row-stochastic source maps, allow one source map to be declared only through
a finite collection of root--word--probe responses.  Every additional source
root is required to be fully specified as an explicitly declared rational
distribution over the calibrated payload basis.

Apply the total parser in \eqref{eq:lics-parser-table} and the compiled
declaration ledger (D0)--(D4), equivalently
\eqref{eq:boundary-prefix-closure}--\eqref{eq:boundary-phase-targets} together
with complete basis-identification queries and copied source tests. Put
\(q:=|Q|\) and \(K:=K_0q\).  Then:
\begin{enumerate}[label=(\roman*)]
\item a source realization on the \(K_0\) payload basis states exists if
and only if the compiled instance admits an exact shared realization with at most \(K\) states;
\item every fully calibrated specification containing all calibrated roots
$h_{z,b}$ for $z\in Z$ and $b\in Q$ and whose declared source schedules are
\emph{legal compound controls} has minimum at least \(K\).  If every source specification separately admits
a \(K_0\)-state realization, then every compiled specification has exact local minimum
\(K\), and the query-specific static union has exact minimum \(K\);
\item legal compound controls are concatenations of complete macros, with
\(\varepsilon\) as the zero-gate control.  Terminal probe programs
\(\pl{P}d_y\#\) and nonempty phase suffixes
\(\pl{T},\ldots,\pl{T}^{L_Q-1}\) serve as calibration queries; proper code
prefixes and malformed strings receive the total default semantics of
\eqref{eq:lics-parser-table}.  The empty word also participates in initial
phase calibration.
\end{enumerate}
The construction has polynomial encoding length: it contains polynomially many
queries, each word has polynomial length, and each rational target has
polynomial bit length.
\end{lemma}

\begin{proof}
For the forward implication, use the code branch \(\pl{P}c_i\#\) for each
source map and the probe branch \(\pl{P}d_y\#\) for each \(e_y\).  Tensor a
\(K_0\)-state source realization with the parser.  Every completed macro
resets the code coordinate and preserves the phase coordinate, so after
every nonempty complete-macro prefix \(v\) the controller state is
\(\operatorname{reset}(b)\).  The product realization therefore satisfies
\eqref{eq:boundary-phase-targets} and all copied source responses on exactly
\(K_0q\) states.

Conversely, fix an exact compiled realization of dimension \(d\le K_0q\).
Let \(M_{\pl T}\) be its transition for \(\pl T\), let \(e\) be its terminal
effect, and put
\[
 f_r:=M_{\pl T}^{\,r}e,\qquad
 C_b:=\{x:f_r(x)=\chi_r(b)\ (0\le r<L_Q)\}.
\]
If a probability row has expectation \(0\), respectively \(1\), against
an effect in \([0,1]^d\), its support lies where that effect equals \(0\),
respectively \(1\).  The \(v=\varepsilon\) cases of
\eqref{eq:boundary-phase-targets} therefore place every root \(h_{g,b}\)
inside \(C_b\).  Injectivity of the phase signatures makes the cells
\(C_b\), \(b\in Q\), pairwise disjoint.

Fix \(b\).  Restrict the \(K_0\) calibrated root rows to \(C_b\), obtaining
a row-stochastic matrix \(W_b\), and restrict the \(K_0\) pulled-back
payload-probe effects \(M_{\pl{P}d_y\#}e\) to \(C_b\), obtaining \(E_b\).
Equation~\eqref{eq:identity-calibration} gives
\[
                         W_bE_b=I_{K_0}.
\]
Hence \(|C_b|\ge K_0\).  Since the \(q\) cells are disjoint and
\(d\le K_0q\),
\[
 d=K_0q,\qquad |C_b|=K_0\quad(b\in Q),
\]
and these cells exhaust the latent carrier.  Thus \(W_b,E_b\) are square
nonnegative inverse matrices.  To see directly that they are monomial, for
\(i\ne j\) write
\[
 0=(W_bE_b)_{ij}=\sum_{\xi\in C_b}(W_b)_{i\xi}(E_b)_{\xi j}.
\]
Hence every coordinate supported by row \(i\) of \(W_b\) has zero \(E_b\)-entry
in every column \(j\ne i\).  The \(K_0\) nonempty row supports of \(W_b\) are
therefore pairwise disjoint---an intersection would give a zero row of the
invertible matrix \(E_b\).  They lie inside the \(K_0\)-element set \(C_b\), so
each is a singleton.  Row normalization makes \(W_b\) a permutation matrix, and its
inverse \(E_b\) is the inverse permutation matrix.  After relabeling, the calibrated roots
in \(C_b\) are point masses \(\delta_{x_{z,b}}\), and the pulled-back payload
probes are their coordinate indicators.

We now identify the payload map induced by each complete macro.
Let \(\omega_i\) be any macro named or used by the specification.  Its singleton
belongs to \(\operatorname{BPref}(\mathcal U)\).  Therefore, for every
\(z,b,r\), the query \(\omega_i\pl T^r\) is declared and
\eqref{eq:boundary-phase-targets} gives
\[
       \delta_{x_{z,b}}M_{\omega_i}
          \bigl(C_{\operatorname{reset}(b)}\bigr)=1 .
\]
As \(z\) ranges over \(Z\), the states \(x_{z,b}\) exhaust \(C_b\).  Hence
the complete macro \(M_{\omega_i}\), although its primitive letters may
visit arbitrary latent states internally, maps every row of \(C_b\)
entirely into \(C_{\operatorname{reset}(b)}\).  It induces a row-stochastic
payload matrix
\[
 G_{i,b}(z,y):=
 M_{\omega_i}(x_{z,b},x_{y,\operatorname{reset}(b)}).
\]
Post-macro payload basis-identification queries read these entries because the pulled-back
probe \(e_y\) is the indicator of
\(x_{y,\operatorname{reset}(b)}\). Complete basis identification therefore fixes
\(G_{i,b}\) to the declared rational source map.  For the partial map, the
copied source tests specify its declared root--word--probe equations.

Choose any synchronized sector
\(b_\phi=(\phi,\varepsilon_{\rm code})\), for which
\(\operatorname{reset}(b_\phi)=b_\phi\).  All restrictions
\(G_{i,b_\phi}\) act on one common \(K_0\)-state payload basis.  Since every
source response was copied into this sector, a schedule
\(\omega_{i_1}\cdots\omega_{i_t}\) has recovered payload action
\[
                  G_{i_1,b_\phi}\cdots G_{i_t,b_\phi}.
\]
Every additional source root is likewise recovered as a probability row on
this basis.  The copied root--word--probe equations are precisely the
original source equations, giving a shared \(K_0\)-state source realization.

The labelled effects give the static lower bound directly.  Fix one
fully calibrated specification label $\lambda$.  In a $d$-state static realization,
let $f_y\in[0,1]^d$ be the effect assigned to the payload query $\zeta_y$ and
let $g_r\in[0,1]^d$ be the effect assigned to $\pl T^r$, all carrying the same label
$\lambda$, and write $p_{z,b}$ for the row of root $h_{z,b}$.  For
$z\in Z$ and $b\in Q$, define
\[
\begin{aligned}
 C_{z,b}:=\{\xi\in[d]:{}&
  f_y(\xi)=\mathbf1\{z=y\}\quad(y\in Z),\\[-0.2em]
 &g_r(\xi)=\chi_r(b)\quad(0\le r<L_Q)\}.
\end{aligned}
\]
The calibration equations give
\[
 p_{z,b}f_y=\mathbf1\{z=y\},
 \qquad p_{z,b}g_r=\chi_r(b).
\]
Because every effect coordinate lies in $[0,1]$, an expectation equal to zero
or one forces $\operatorname{supp}(p_{z,b})\subseteq C_{z,b}$.  The binary
fingerprints
\[
 \bigl((\mathbf1\{z=y\})_{y\in Z},\chi(b)\bigr)
\]
are pairwise distinct.  Hence the $K_0q$ nonempty sets $C_{z,b}$ are pairwise
disjoint.  The independent labelled effects therefore give the static lower
bound $d\ge K_0q=K$; the same endpoint argument with pulled-back
shared effects gives the local lower bound.

For the upper bounds, choose a $K_0$-state source realization separately for
each specification and tensor it with the parser. Identity and root queries let us
relabel these local witnesses so that their compiled roots coincide: the
calibrated roots are $\delta_{(z,b)}$, and every additional root is its
declared mixture in sector $b$.  Retain these common roots and, for every
labelled query $(\lambda,u)$, use the pulled-back effect
$M_u^{(\lambda)}e^{(\lambda)}$ from that specification's local witness. Independence
of labelled effects gives a $K$-state static realization of the union, while the
tensor products give the matching local upper bounds. This is exactly where
the labelled relaxation is used: if the same word occurs under different labels,
the corresponding local witnesses may produce different effect vectors even
when the declared scalar targets agree. An unlabelled static relaxation would
require one common vector $e_u$ and would need a separate upper-bound
argument.

Boundary phase words serve as queries attached to legal macro sequences.  Let
\(H_{\mathcal U}\) be the source roots used by a specification.  If that specification has
\(s\) scheduled macro occurrences, then
\[
 |\operatorname{BPref}(\mathcal U)|
 \le 1+|\{i:\omega_i\text{ is used}\}|+s .
\]
The boundary block therefore adds
\[
 O\!\left(|H_{\mathcal U}|\,qL_Q\,
          |\operatorname{BPref}(\mathcal U)|\right)
\]
queries of polynomial length with \(0/1\) targets.  Every declared trace
follows a valid parser branch; the total default rows complete the remaining
transition rows.
\end{proof}

\begin{remark}[Witness construction and converse reconstruction]
\label{rem:black-box}
The parser table supplies the tensor-product witness in the forward direction.
In the converse, the declared response equations reconstruct the required
structure in every realization of dimension $d\le K$: initial phase signatures
identify the sector cells, the identity table identifies each cell with the
payload basis, singleton boundary signatures map complete macros to their
prescribed reset cells, and post-macro payload probes recover the induced
source maps.  Since the $q$ cells each contain $K_0$ states and
$d\le K_0q$, they exhaust the carrier at dimension $K=K_0q$; consecutive
source gates therefore act through one synchronized payload basis.  The
leakage bounds provide the quantitative version of this reconstruction.
Every declared word follows a valid branch, and the default parser rows
complete the stochastic transition table on the remaining inputs.
\end{remark}

\begin{corollary}[Preservation of the two calibrated hard families]
\label{cor:identity-hard-family-preservation}
Lemma~\ref{lem:identity-fixed-interface} applies both to the ETR--INV gate
source below and to the robust Boolean source below, with
\(q=|Q|\) and \(K=K_0q\). Their specifications are fully calibrated, use partial
tests for the single source map $V$, and use complete basis-identification queries for all
explicitly rational fixed gates and every additional rational root.  The lemma
preserves exact shared realization and both declared thresholds below.
In the ETR
family, each equation is separately feasible; in the Boolean family, each
Booleanity or clause test is separately satisfiable, so both families meet the
local-feasibility hypothesis of Lemma~\ref{lem:identity-fixed-interface}(ii).
The corresponding
quantitative robust bounds for the compiled instances are established in
Lemmas~\ref{lem:approx-sector-allocation}--%
\ref{lem:boolean-clause-extraction}.
\end{corollary}
\begin{proof}
Sections~\ref{app:etr-source} and~\ref{app:robust-boolean} verify the stated
basis-identification and separate-feasibility hypotheses for the two source families.
Applying Lemma~\ref{lem:identity-fixed-interface} gives the exact shared,
local and query-specific conclusions. The quantitative conclusion follows
from Lemmas~\ref{lem:approx-sector-allocation}--%
\ref{lem:boolean-clause-extraction}, which reconstruct the synchronized sector
maps from the declared responses.
\end{proof}

\subsection{The identity-constrained ETR--INV branch}
\label{app:etr-source}

Use ETR--INV in the bounded rational normal form of variables
$x_1,\ldots,x_n\in[1/2,2]$ and atoms $x_i=1$, $x_i+x_j=x_k$, and
$x_ix_j=1$; this form is $\exists\mathbb{R}$-complete by
\citep[Theorem~6]{abrahamsen2018artgallery}.  For every unit atom
$\alpha:x_i=1$, introduce a fresh variable $y_\alpha\in[1/2,2]$ and replace
$\alpha$ by the two atoms
\[
                         x_iy_\alpha=1,
                         \qquad x_i=y_\alpha .
\]
This is equivalent: a solution with $x_i=1$ extends by $y_\alpha=1$, while
conversely $x_i=y_\alpha>0$ and $x_iy_\alpha=1$ imply $x_i^2=1$ and hence
$x_i=1$.  Thus the inversion atom uses distinct variable symbols.  Replace
every appearance of every variable, including each $y_\alpha$, by a fresh
occurrence variable.  For each underlying variable, connect its occurrences
by a spanning tree of copy atoms.  Every addition or inversion atom now has
distinct occurrence symbols, while the copy atoms force all occurrences of
an underlying variable to agree.  If the constraint list is empty, append an
always-feasible copy atom on two fresh variables, producing a nonempty specification
family.  The postprocessed instance is polynomial, uses only
addition, inversion, and copy atoms, and is feasible exactly when the original
instance is.  Set
\(x_i=\frac12+\frac32t_i\), \(t_i\in[0,1]\).  Direct expansion gives the
equivalent acceptance equations
\begin{align}
 \frac{t_i+t_j+1-t_k}{3}&=\frac29, &
 \frac{t_i+t_j+3t_it_j}{5}&=\frac15, &
 \frac{t_i+1-t_j}{2}&=\frac12 .
\label{eq:etr-acceptance}
\end{align}
For example,
\((\frac12+\frac32t_i)(\frac12+\frac32t_j)=1\) is precisely
\(t_i+t_j+3t_it_j=1\).

For every postprocessed equation $c$, fix an ordered list of branch
occurrences $\mathcal B_c$ and their rational weights:
\begin{equation}
\begin{array}{c|c|c}
c&\mathcal B_c&(\alpha_\beta)_{\beta\in\mathcal B_c}\\ \hline
x_i+x_j=x_k&(\mathsf L_i^+,\mathsf L_j^+,\mathsf L_k^-)
 &(1/3,1/3,1/3)\\
x_ix_j=1&(\mathsf L_i^+,\mathsf L_j^+,\mathsf P_{ij})
 &(1/5,1/5,3/5)\\
x_i=x_j&(\mathsf L_i^+,\mathsf L_j^-)&(1/2,1/2).
\end{array}
\label{eq:etr-branch-list}
\end{equation}
Repeated branch types label distinct occurrences.  For each
such occurrence $\beta$, introduce a waiting state $w_\beta$ and parked states
$a_\beta^\circ,r_\beta^\circ$.  With
$U:=\{u_i:i\text{ is a split occurrence}\}$, put
\begin{equation}
 Z=U\mathbin{\dot\cup}\{s,d,a,r\}
   \mathbin{\dot\cup}
   \{w_\beta,a_\beta^\circ,r_\beta^\circ:
                         c,\ \beta\in\mathcal B_c\}.
\label{eq:etr-state-set}
\end{equation}
The only partially declared gate is $V$.  On every $u_i$-row all probes
outside $\{s,d\}$ are declared zero, so exact identity constraints and row
stochasticity give
\begin{equation}
 \delta_{u_i}V=t_i\delta_s+(1-t_i)\delta_d,
 \qquad t_i\in[0,1].
\label{eq:bernoulli-V}
\end{equation}
The map $V$ fixes $a,r$ and every waiting or parked state; its $s,d$ rows are
undeclared.

All remaining gates are deterministic and completely specified. The table
lists every nonidentity row; every omitted row is an identity row.
\begin{equation}
\begin{array}{c|l}
\text{gate}&\text{nonidentity rows}\\ \hline
\mathsf A_\beta&w_\beta\mapsto u_{i(\beta)}\\
\mathsf R_j&s\mapsto u_j,\quad d\mapsto r\\
\mathsf F^+&s\mapsto a,\quad d\mapsto r\\
\mathsf F^-&s\mapsto r,\quad d\mapsto a\\
\mathsf\Pi_\beta&a\mapsto a_\beta^\circ,\quad
                         r\mapsto r_\beta^\circ\\
\mathsf C_c&a_\beta^\circ\mapsto a,\quad
             r_\beta^\circ\mapsto r\quad(\beta\in\mathcal B_c).
\end{array}
\label{eq:etr-deterministic-gates}
\end{equation}
Here $i(\mathsf L_i^\pm)=i(\mathsf P_{ij})=i$.  For a branch occurrence
$\beta$, define its complete word by
\begin{align}
 \eta_\beta&:=\mathsf A_\beta V\mathsf F^+\mathsf\Pi_\beta
     &&(\beta\text{ of type }\mathsf L_i^+),\nonumber\\
 \eta_\beta&:=\mathsf A_\beta V\mathsf F^-\mathsf\Pi_\beta
     &&(\beta\text{ of type }\mathsf L_i^-),\nonumber\\
 \eta_\beta&:=\mathsf A_\beta V\mathsf R_jV
                    \mathsf F^+\mathsf\Pi_\beta
     &&(\beta\text{ of type }\mathsf P_{ij}).
\label{eq:etr-branch-words}
\end{align}
If $\mathcal B_c=(\beta_1,\ldots,\beta_m)$, use the fully specified
initial-distribution root and the complete test word
\begin{equation}
 p_{h_c}:=\sum_{\beta\in\mathcal B_c}
              \alpha_\beta\delta_{w_\beta},
 \qquad
 v_c:=\eta_{\beta_1}\cdots\eta_{\beta_m}\mathsf C_c.
\label{eq:etr-root-and-word}
\end{equation}

\begin{lemma}[Routing/parking evaluation]
\label{lem:etr-routing-parking}
For each equation, the transition product induced by the schedule in
\eqref{eq:etr-root-and-word} is row-stochastic, presents \(V\) only with
distributions supported outside \(\{s,d\}\), and has accept probability equal to the corresponding left side of
\eqref{eq:etr-acceptance}.
\end{lemma}
\begin{proof}
Induct over the ordered branch occurrences.  Before a branch, completed mass
is in its own parked states, unprocessed mass is in its own waiting states,
and no mass is on $s$ or $d$.  Activation moves only the current waiting mass
to its variable state.  An immediate routing removes the mass that $V$ puts
on $s,d$; in a product branch $\mathsf R_j$ leaves the surviving mass on
$u_j$ and the rejected mass on $r$, so the second $V$ also receives no mass on
$s,d$.  Parking restores the invariant.  Thus the three branch types accept
with probabilities $t_i$, $1-t_i$, and $t_it_j$, respectively.  Linearity and
\eqref{eq:etr-branch-list} give the three left sides of
\eqref{eq:etr-acceptance}; $\mathsf C_c$ then collects all successful parked
mass in $a$.  Every gate is total and row stochastic because its omitted rows
are identity rows.  Each equation has at most three branches, so its word has
constant length.
\end{proof}

Let the root set contain the calibrated roots $h_z$ for $z\in Z$ and every
initial-distribution root $h_c$ in \eqref{eq:etr-root-and-word}. Each $h_c$ is a named
probability distribution on the common payload carrier.  Let $\mathcal S_{\rm ETR}$ contain
the identity constraints, complete basis identification of every $h_c$ and every gate
in \eqref{eq:etr-deterministic-gates}, and the displayed partial basis identification of
$V$.  For each postprocessed equation $c$, let
\[
 \mathcal U_c:=\mathcal S_{\rm ETR}\cup
       \{p_{h_c}G_{v_c}e_a=r_c\},
\]
where \(v_c\) is its evaluation schedule and
\(r_c\in\{2/9,1/5,1/2\}\) is the corresponding target in
\eqref{eq:etr-acceptance}.  After compilation, enlarge \(\mathcal U_c\) by
the boundary closure
\eqref{eq:boundary-prefix-closure}--\eqref{eq:boundary-phase-targets} for
every calibrated or initial-distribution root used by that specification. In particular, the
reset-aware phase signature is declared after each gate of \(v_c\) and after
each gate used as a singleton; its targets are independent of the unknown
tuple \(t\).  Thus the reduction is a finite rational specification, and its shared
instance is \(\bigcup_c\mathcal U_c\).

\begin{theorem}[Exact calibrated gate reduction]
\label{thm:exact-gate-hardness}
Let \(K_0=|Z|\).  For every input, each individual equation specification and the
query-specific static union have exact minimum \(K_0\).  A single
\(K_0\)-state model shared by all specifications exists if and only if the
ETR--INV instance is feasible.  After
Lemma~\ref{lem:identity-fixed-interface}, the same statements hold with
\(K=K_0q\) over the fixed five-letter one-effect interface.
\end{theorem}
\begin{proof}
Identity calibration gives the lower bound and identifies every
\(K_0\)-state model with the basis.  Lemma~\ref{lem:etr-routing-parking}
then shows that one shared matrix \(V\) realizes all equation targets exactly
if and only if its common parameters \(t_i\) satisfy all split equations.
Each equation alone is feasible: addition uses
\((x_i,x_j,x_k)=(1/2,1/2,1)\), inversion uses \((1,1)\), and a copy
equation uses equal values.  Choosing such values separately supplies every
local upper bound.  For the static union, use the calibrated roots as the
$K_0$ basis states and every initial-distribution root as its declared rational
mixture.  For each labelled response column, take the pulled-back effect from
that specification's local witness; linearity makes its values on all initial-distribution roots
consistent with the declared mixtures.  The fixed
interface conclusion is Corollary~\ref{cor:identity-hard-family-preservation}.
\end{proof}

\subsection{The robust Boolean branch}
\label{app:robust-boolean}

Start from $(3,4)$-SAT~\citep{tovey1984simplified}.  Delete isolated
variables and, if there is no clause, conjoin one clause
$z_1\vee z_2\vee z_3$ on fresh variables.  This preserves satisfiability and
the occurrence bound.  Canonically rename the remaining variables
consecutively and put
\[
 N_{\rm bool}:=\max\{2,p+q_F\},
\]
where $p$ is the number of variables and $q_F$ the number of clauses after
this preprocessing.  We use this combinatorial size throughout.  Its standard
explicit binary encoding length lies between $\Omega(N_{\rm bool})$ and
$O(N_{\rm bool}\log N_{\rm bool})$, and is therefore polynomially equivalent
to $N_{\rm bool}$.  Use identity constraints on
\[
 Z:=\{u_1,\ldots,u_p\}\mathbin{\dot\cup}\{s,d,a,r\},
 \qquad K_0=p+4,
\]
and put
\[
 Z_{\rm var}:=\{u_1,\ldots,u_p\},\qquad
 Z_{\rm tr}:=\{s,d\},\qquad
 Z_{\rm fix}:=Z\setminus(Z_{\rm var}\cup Z_{\rm tr}).
\]
The single partially declared source map \(V\) is specified as follows.  For
every \(i\) and every \(y\notin\{s,d\}\), declare
\[
                         R(h_{u_i},Ve_y)=0.
\]
The input leaves the two responses against \(e_s,e_d\) as completion
variables.  For every
\(z\in Z_{\rm fix}\), completely declare the row
\[
                         R(h_z,Ve_y)=\mathbf1\{z=y\}
                         \qquad(y\in Z).
\]
Thus \(V\) fixes $a,r$.  Ideal schedules apply $V$ only to variable states and
$r$, while its $s,d$ rows remain completion variables.  In
every exact \(K_0\)-state identity-constrained source realization,
row-stochasticity consequently gives parameters \(b_i\in[0,1]\) satisfying
\begin{equation}
                        \delta_{u_i}V
                        =b_i\delta_s+(1-b_i)\delta_d .
\label{eq:boolean-V}
\end{equation}
For every $j$, introduce the completely specified deterministic gates
$R_j^s,R_j^d$, together with the two global gates $F^s,F^d$.  The following
table lists their $s,d$ rows; every omitted row is an identity row.
\begin{equation}
\begin{array}{c|cc}
 &s&d\\ \hline
R_j^s&u_j&r\\
R_j^d&r&u_j\\
F^s&a&r\\
F^d&r&a
\end{array}
\label{eq:boolean-routes}
\end{equation}
These rows specify total row-stochastic maps.  The Booleanity word for $x_i$
is
\begin{equation}
                  v_i^{\rm bool}:=VR_i^sVF^d .
\label{eq:booleanity-word}
\end{equation}
From $h_{u_i}$ its response against $e_a$ is $b_i(1-b_i)$.

For a literal $\ell$, put $f(x_i):=d$ and $f(\neg x_i):=s$.  If
$C=(\ell_1\vee\ell_2\vee\ell_3)$ and $\ell_t$ uses variable $x_{i_t}$,
define the clause word
\begin{equation}
 v_C:=V R_{i_2}^{f(\ell_1)}
        V R_{i_3}^{f(\ell_2)}
        V F^{f(\ell_3)} .
\label{eq:boolean-clause-word}
\end{equation}
At each routing step, failure mass moves to the next variable and all
complementary mass moves to $r$.  Thus, if $\nu_j$ is the surviving failure
mass after the first $j$ literals, then $\nu_0=1$ and
$\nu_j=\nu_{j-1}(1-\ell_j)$.  Hence the clause response is
\begin{equation}
 p_{h_{u_{i_1}}}G_{v_C}e_a
 =\nu_3=\prod_{j=1}^3(1-\ell_j).
\label{eq:clause-product}
\end{equation}
Immediately before every occurrence of $V$, the distribution is supported on
the next variable state together with $r$.  Because $V$ fixes $r$, the
declared variable and reject rows determine the complete tested evolution.  A
Booleanity word has four gates and a clause word has six, so set $g_0:=6$.

Let \(\mathcal S_{\rm bool}\) contain the identity constraints, the fixed
rows of \(V\) declared above, complete basis identification of all gates in
\eqref{eq:boolean-routes}, and the zero probes defining the variable rows of
\(V\).  Define
\begin{align}
 \mathcal U_i^{\rm bool}
  &:=\mathcal S_{\rm bool}\cup
     \{p_{h_{u_i}}G_{v_i^{\rm bool}}e_a=0\},\nonumber\\
 \mathcal U_C
  &:=\mathcal S_{\rm bool}\cup
     \{p_{h_{u_{i_1}}}G_{v_C}e_a=0\}.
\label{eq:boolean-source-specifications}
\end{align}
Every such specification is made boundary-complete in the sense of
\eqref{eq:boundary-prefix-closure}--\eqref{eq:boundary-phase-targets}.
In particular, every gate used in a schedule occurs as a singleton macro,
and the reset-aware phase signature is declared after every scheduled gate
prefix from every calibrated root required by that specification.  These targets are
independent of the free Boolean parameters.  The total source family consists
of the specifications in \eqref{eq:boolean-source-specifications} and has size polynomial in
$N_{\rm bool}$.  A Booleanity specification is feasible by choosing
$b_i\in\{0,1\}$; a clause specification is feasible by making one of its literals true,
with all other values arbitrary.  Each specification contains its own Booleanity or
clause test together with the common calibration family, so its free
parameters may be chosen independently.

Compile this source by Lemma~\ref{lem:identity-fixed-interface}.  Let \(L_Q\)
be the de Bruijn signature length, \(q=|Q|\), and \(K=K_0q\).

\begin{lemma}[Approximate sector allocation and phase leakage]
\label{lem:approx-sector-allocation}
Suppose a compiled realization has \(d\le K\) and pointwise defect
\(\dr\).  Put \(e_0:=(2+8L_Q)\dr\).  If \(K_0e_0<1\), rounding the
\(L_Q\) pulled-back phase effects at \(1/2\), with ties assigned to $1$,
defines \(q\) cells that partition the latent carrier.  Each
calibrated root has mass at most \(2L_Q\dr\) outside its prescribed cell,
every cell contains at least \(K_0\) latent states, and consequently
\(d=K\) and every cell contains exactly \(K_0\) states.  In particular,
when \(d<K\), defect \(\dr<1/(K_0(2+8L_Q))\) is impossible.
\end{lemma}
\begin{proof}
Put $f_r:=M_{\pl T}^{r}e$.  For a latent state $\xi$, define
\[
 \widehat\chi_r(\xi):=
 \begin{cases}
  0,&f_r(\xi)<1/2,\\
  1,&f_r(\xi)\ge1/2,
 \end{cases}
\]
Set
\[
 \widehat\chi(\xi):=
   (\widehat\chi_0(\xi),\ldots,\widehat\chi_{L_Q-1}(\xi)),
 \qquad
 C_b:=\{\xi\in[d]:\widehat\chi(\xi)=\chi(b)\}.
\]
By \eqref{eq:lics-phase-signature}, $q=2^{L_Q}$ and
$\chi:Q\to\{0,1\}^{L_Q}$ is a bijection.  Hence every $\xi\in[d]$ belongs to
exactly one $C_b$, and the $q$ cells form a disjoint partition of $[d]$.
For a target phase bit zero or one, a state carrying the wrong rounded bit
contributes at least half its mass to the response error.  A union bound over
the \(L_Q\) bits therefore leaves at most \(2L_Q\dr\) mass of each calibrated
root outside its prescribed cell.
Fix a phase sector.  It has exactly \(K_0\) calibrated roots, one for each
payload index.  If the sector contains \(d_\tau\) latent states, restricting
and renormalizing those roots gives a row-stochastic matrix
$W\in[0,1]^{K_0\times d_\tau}$, while restricting the \(K_0\) payload effects
gives $E\in[0,1]^{d_\tau\times K_0}$.  Thus $WE$ is always a
$K_0\times K_0$ matrix. If $o\le2L_Q\dr$ is the removed mass in one row, then for
every payload column the omitted contribution is in $[0,o]$, and
renormalization changes the retained contribution by at most
$(\dr+o)/(1-o)$.  Under $K_0e_0<1$ we have $o<1/2$, so
\begin{equation}
 \bigl|(WE)_{zy}-\mathbf 1\{z=y\}\bigr|
 \le 2\dr+4o\le(2+8L_Q)\dr=e_0 .
\end{equation}
Thus $\|WE-I\|_2\le K_0\|WE-I\|_{\max}<1$.  If the cell had fewer
than $K_0$ states, then $\operatorname{rank}(WE)\le d_\tau<K_0$, whereas
\(\sigma_{\min}(I+(WE-I))\ge1-\|WE-I\|_2>0\).
Thus all \(q\) cells contain at least \(K_0\) states.  Since
\(d\le K_0q\), equality holds everywhere, proving also the \(d<K\) claim.
\end{proof}

\begin{lemma}[Approximate identity/basis rigidity]
\label{lem:approx-basis-rigidity}
Let \(W\) be \(K_0\times K_0\) row stochastic, let
\(E\in[0,1]^{K_0\times K_0}\), and suppose
\(\|WE-I\|_{\max}\le e_0\le1/(4(K_0+1))\).  There is a permutation
\(\pi\) such that, for \(i\ne j\),
\begin{equation}
 W_{i,\pi(i)}\ge1-2(K_0+1)e_0,\quad
 E_{\pi(i),i}\ge1-(2K_0+3)e_0,\quad E_{\pi(i),j}\le2e_0 .
\label{eq:approx-basis}
\end{equation}
\end{lemma}
\begin{proof}
Put $G_i=\{z:E_{zi}>1/2\}$.  The diagonal inequality gives row $i$
mass at least $1-2e_0$ on $G_i$.  If $z\in G_i\cap G_j$ with $i\ne j$,
then $W_{iz}E_{zj}\le(WE)_{ij}\le e_0$ and $E_{zj}>1/2$, so
$W_{iz}\le2e_0$.  Hence row $i$ has mass at most $2K_0e_0$
on all intersections, and
\[
 W_i\bigl(G_i\setminus\bigcup_{j\ne i}G_j\bigr)
 \ge1-2(K_0+1)e_0>0.
\]
The sets $G_i':=G_i\setminus\bigcup_{j\ne i}G_j$ are pairwise disjoint
and there are $K_0$ of them inside a set of $K_0$ latent states.  Each is
therefore a singleton, defining a permutation $\pi$.  The displayed mass
bound gives the estimate for $W_{i,\pi(i)}$.  Substitution in the diagonal
entry gives the estimate for $E_{\pi(i),i}$, and substitution in the
off-diagonal entry gives $E_{\pi(i),j}\le2e_0$.
\end{proof}

\begin{lemma}[Gate transfer by total-variation contraction]
\label{lem:robust-boolean-gate}
Set
\[
 c_0:=512,\qquad \kappa:=c_0(1+K_0)^2(1+L_Q).
\]
Every compiled realization of dimension $d\le K$ and pointwise defect
\(\dr\le1/\kappa\)
necessarily has $d=K$.  Fix any synchronized sector
$b_\phi=(\phi,\varepsilon_{\rm code})$, so that
$\operatorname{reset}(b_\phi)=b_\phi$, and let $C$ be the rounded cell whose
phase signature is $\chi(b_\phi)$.  Then $|C|=K_0$, and there is a bijection
$\pi:Z\to C$ for which the following assertions hold.  Put
\[
\begin{aligned}
 e_y^{\rm pull}&:=M_{\zeta_y}e,&
 \operatorname{dec}(x)_y&:=xe_y^{\rm pull},\\
 x_\pi&:=(x_{\pi(z)})_{z\in Z},&
 \theta(x)&:=x(C^c).
\end{aligned}
\]
For every completely declared source gate
\(G\), every latent probability row \(x\), and the corresponding candidate
macro \(M_G\),
\begin{equation}
\begin{split}
 \|\operatorname{dec}(xM_G)-\operatorname{dec}(x)G\|_1
       &\le\kappa(\dr+\theta(x)),\\
 \theta(xM_G)&\le\theta(x)+\kappa\dr .
\end{split}
\label{eq:routing-transfer}
\end{equation}

There are recovered values \(\widehat b_i\in[0,1]\) and a row-stochastic
source map \(\widehat V\) such that
\[
 \delta_{u_i}\widehat V
   =\widehat b_i\delta_s+(1-\widehat b_i)\delta_d,
\]
\(\widehat V\) agrees with the declared fixed rows of \(V\), and its rows on
\(s,d\) may be chosen arbitrarily.  With
\[
 \mu(x):=x_{\pi(s)}+x_{\pi(d)},
\]
one has
\begin{equation}
\begin{split}
 \|\operatorname{dec}(xM_V)
       -\operatorname{dec}(x)\widehat V\|_1
   &\le\kappa(\dr+\theta(x))+(K_0+1)\mu(x),\\
 \theta(xM_V)
   &\le\theta(x)+\kappa\dr+\mu(x).
\end{split}
\label{eq:partial-v-transfer}
\end{equation}

For every declared Booleanity or clause schedule \(w\) with at most \(g_0\)
source gates, let \(P_w(\widehat b)\) be its ideal source acceptance
probability after replacing \(V\) by \(\widehat V\).  Then
\begin{equation}
 \left|\operatorname{resp}(w)-P_w(\widehat b)\right|
 \le C_0\dr,\qquad
 C_0:=11(\kappa+10)^{g_0+1}.
\label{eq:word-transfer}
\end{equation}
\end{lemma}

\begin{proof}
For readability put
\[
\begin{aligned}
 n&:=K_0,& L&:=L_Q,\\
 A_0&:=(n+1)^2(1+L),& e_0&:=(2+8L)\dr,
\end{aligned}
\]
 so that $\kappa=512A_0$, and put
\[
\begin{aligned}
 \alpha&:=6(n+1)e_0,& \beta&:=(2n+3)e_0,\\
 \eta&:=n(2n+3)e_0,& \lambda&:=2L\dr .
\end{aligned}
\]
Since $e_0\le8(1+L)\dr$ and $\dr\le1/\kappa$, the fixed choice
$c_0=512$ gives
\[
\begin{aligned}
 e_0&\le\frac1{64(n+1)^2}\le\frac1{4(n+1)},\\
 ne_0&\le\frac{n}{64(n+1)^2}\le\frac1{256}<1,\\
 \beta&\le\frac{24}{512(n+1)}\le\frac3{128}<\frac12 .
\end{aligned}
\]
Lemma~\ref{lem:approx-sector-allocation} therefore forces $d=K$ and
$|C|=n$.  Write $p_z$ for the candidate calibrated root associated with
$z\in Z$ in sector $b_\phi$.  Let $W$ have as row $z$ the normalized
restriction of $p_z$ to $C$, and define
\[
 E_{\xi,y}:=e_y^{\rm pull}(\xi)\qquad(\xi\in C,\ y\in Z).
\]
The first displayed smallness bound verifies the hypothesis of
Lemma~\ref{lem:approx-basis-rigidity}; applied to $W,E$, that lemma supplies
the asserted bijection $\pi:Z\to C$.
We shall use the explicit bounds
\begin{equation}
 \alpha\le48(n+1)(1+L)\dr,\qquad
 \eta\le24A_0\dr,\qquad
 \lambda\le2(1+L)\dr .
\label{eq:robust-numerical-bounds}
\end{equation}

Let $o_z:=\theta(p_z)$.  Initial phase calibration gives $o_z\le\lambda$.
Writing $W_z$ for row $z$ of $W$, basis rigidity gives
$W_z(\pi(z))\ge1-2(n+1)e_0$.  Consequently
\begin{equation}
\begin{split}
 \|p_z-\delta_{\pi(z)}\|_1
 &=2\bigl(1-p_z(\pi(z))\bigr)\\
 &\le2o_z+4(n+1)e_0
 \le2\lambda+4(n+1)e_0\le\alpha,
\end{split}
\label{eq:calibrated-root-to-basis}
\end{equation}
where $\lambda\le(n+1)e_0$.  For the pulled-back payload effects, the row
indexed by $\pi(z)$ differs from the $z$-th coordinate row by at most
\[
 (1-E_{\pi(z),z})+\sum_{y\ne z}E_{\pi(z),y}
 \le\beta+2(n-1)e_0=(4n+1)e_0\le\eta .
\]
Decomposing an arbitrary latent probability row as
\[
 x=\sum_{z\in Z}x_{\pi(z)}\delta_{\pi(z)}+x^\perp,
 \qquad \|x^\perp\|_1=\theta(x),
\]
and using that every pulled-back effect lies in $[0,1]$ gives
\begin{equation}
 \|\operatorname{dec}(x)-x_\pi\|_1
 \le\eta+n\theta(x)\le\eta+(n+1)\theta(x).
\label{eq:decoded-coordinate-error}
\end{equation}

Now let $G$ be completely declared. Complete post-macro payload basis identification
and the boundary phase signature give, for every $z\in Z$,
\begin{equation}
 \|\operatorname{dec}(p_zM_G)-\delta_zG\|_1\le n\dr,
 \qquad \theta(p_zM_G)\le\lambda .
\label{eq:complete-gate-on-calibrated-root}
\end{equation}
Stochastic $\ell_1$-contraction, followed by evaluation against the $n$
pulled-back effects, therefore yields
\begin{equation}
\begin{split}
 \|\operatorname{dec}(\delta_{\pi(z)}M_G)-\delta_zG\|_1
   &\le n\dr+n\alpha,\\
 \theta(\delta_{\pi(z)}M_G)&\le\lambda+\alpha .
\end{split}
\label{eq:complete-gate-on-basis-state}
\end{equation}

For an arbitrary $x$, put $a:=x_\pi$ and
$\rho:=\operatorname{dec}(x)-a$.  Since $G$ is row stochastic,
$\|\rho G\|_1\le\|\rho\|_1$.  The preceding decomposition and
Equations~\eqref{eq:decoded-coordinate-error} and
\eqref{eq:complete-gate-on-basis-state} give
\[
\begin{split}
 &\|\operatorname{dec}(xM_G)-\operatorname{dec}(x)G\|_1\\
 &\quad\le
 \sum_z a_z
 \|\operatorname{dec}(\delta_{\pi(z)}M_G)-\delta_zG\|_1
 +\|\operatorname{dec}(x^\perp M_G)\|_1+\|\rho G\|_1\\
 &\quad\le n\dr+n\alpha+n\theta(x)
              +\eta+(n+1)\theta(x)\\
 &\quad\le73A_0\dr+3A_0\theta(x)
 \le\kappa\bigl(\dr+\theta(x)\bigr).
\end{split}
\]
Here $n\dr\le A_0\dr$, $n\alpha\le48A_0\dr$,
$\eta\le24A_0\dr$, and $2n+1\le3A_0$.  Likewise,
\[
\begin{aligned}
 \theta(xM_G)
 &\le\sum_z a_z\theta(\delta_{\pi(z)}M_G)
       +\|x^\perp M_G\|_1\\
 &\le\theta(x)+\lambda+\alpha
 \le\theta(x)+50A_0\dr
 \le\theta(x)+\kappa\dr .
\end{aligned}
\]
This proves \eqref{eq:routing-transfer} with the stated $\kappa$.

For the partially specified rows of \(V\), put
\[
                         q_i:=p_{u_i}M_V,\qquad
                         \widehat b_i:=\operatorname{dec}(q_i)_s.
\]
The phase signature immediately after the \(V\)-macro gives
\[
                         \theta(q_i)\le\lambda .
\]
For \(y\notin\{s,d\}\), the declared zero probe gives
\(\operatorname{dec}(q_i)_y\le\dr\).  Equation
\eqref{eq:decoded-coordinate-error} also gives
\[
 \|\operatorname{dec}(q_i)-(q_i)_\pi\|_1
 \le\eta+(n+1)\lambda .
\]
Set
$\mathbf d:=\operatorname{dec}(q_i)$, $\mathbf a:=(q_i)_\pi$, and
$T:=Z\setminus\{s,d\}$, and put
\[
 t:=\sum_{y\in T}\mathbf d_y,
 \qquad U:=\sum_{y\in T}\mathbf a_y.
\]
The zero probes give $t\le(n-2)\dr$, while
\[
 U\le t+\sum_{y\in T}|\mathbf d_y-\mathbf a_y|.
\]
Since $\widehat b_i=\mathbf d_s$ and
$\mathbf a_s+\mathbf a_d=1-\theta(q_i)-U$, it follows that
\begin{equation}
\begin{split}
 &\left\|
 \mathbf d-
 \bigl(\widehat b_i\delta_s+(1-\widehat b_i)\delta_d\bigr)
 \right\|_1\\
 &=t+|\mathbf d_s+\mathbf d_d-1|\\
 &\le t+|\mathbf d_s-\mathbf a_s|+|\mathbf d_d-\mathbf a_d|
          +\theta(q_i)+U\\
 &\le2t+\|\mathbf d-\mathbf a\|_1+\theta(q_i)\\
 &\le\eta+(n+2)\lambda+2(n-2)\dr .
\end{split}
\label{eq:partial-v-probability-partition}
\end{equation}
Thus actual latent mass together with basis rigidity determines the
complementary \(s,d\) split.  Passing from \(p_{u_i}\) to
\(\delta_{\pi(u_i)}\) adds at most $n\alpha$ to decoded
$\ell_1$-error and at most $\alpha$ to leakage.  Since $s,d$ are distinct,
$n\ge2$, and \eqref{eq:robust-numerical-bounds} gives
\[
\begin{split}
 &\eta+(n+2)\lambda+2(n-2)\dr+n\alpha\\
 &\qquad\le(24+4+2+48)A_0\dr=78A_0\dr,\\
 &\lambda+\alpha\le50A_0\dr .
\end{split}
\]
Consequently, for every $z\notin\{s,d\}$,
\begin{equation}
\begin{split}
 \|\operatorname{dec}(\delta_{\pi(z)}M_V)
       -\delta_z\widehat V\|_1&\le78A_0\dr,\\
 \theta(\delta_{\pi(z)}M_V)&\le50A_0\dr .
\end{split}
\label{eq:partial-v-good-basis-rows}
\end{equation}
For $z\in Z_{\rm fix}$ these bounds follow from complete basis identification; for
$z=u_i$ they follow from
\eqref{eq:partial-v-probability-partition}.  The rows $s,d$ are undeclared
and their rows of $\widehat V$ may be chosen arbitrarily.

For a general $x$, write $a:=x_\pi$,
$\rho:=\operatorname{dec}(x)-a$, and
\[
 x=\sum_z a_z\delta_{\pi(z)}+x^\perp,
 \qquad \|x^\perp\|_1=\theta(x).
\]
The total coefficient of the two exceptional basis rows is
$\mu(x)=a_s+a_d$.  On either such row the decoded vector has coordinates in
$[0,1]$, while the corresponding row of $\widehat V$ is a probability row,
so the decoded discrepancy is at most $n+1$ and its leakage is at most one.
Since $\|\rho\widehat V\|_1\le\|\rho\|_1$, Equations
\eqref{eq:decoded-coordinate-error} and
\eqref{eq:partial-v-good-basis-rows} imply
\[
\begin{split}
 &\|\operatorname{dec}(xM_V)
       -\operatorname{dec}(x)\widehat V\|_1\\
 &\quad\le78A_0\dr+(n+1)\mu(x)+n\theta(x)
                 +\eta+(n+1)\theta(x)\\
 &\quad\le102A_0\dr+(2n+1)\theta(x)+(n+1)\mu(x)\\
 &\quad\le\kappa\bigl(\dr+\theta(x)\bigr)+(n+1)\mu(x),
\end{split}
\]
and
\[
 \theta(xM_V)
 \le50A_0\dr+\mu(x)+\theta(x)
 \le\theta(x)+\kappa\dr+\mu(x).
\]
These are exactly the inequalities in \eqref{eq:partial-v-transfer}.  The
$n$ decoded effects contribute the factor $n+1$.

Fix a declared schedule and let \(x_r\) be its candidate latent row after
\(r\) gates.  Let \(v_r\) be the ideal source row obtained from the declared
fixed gates and \(\widehat V\).  Define
\[
 E_r:=\|\operatorname{dec}(x_r)-v_r\|_1,\qquad
 \theta_r:=\theta(x_r),\qquad
 S_r:=E_r+\theta_r .
\]
Retaining the unsimplified coefficients established above, stochastic
contraction gives, respectively for a complete gate and for a $V$-gate,
\[
\begin{aligned}
 S_r&\le S_{r-1}+123A_0\dr+3A_0\theta_{r-1},\\
 S_r&\le S_{r-1}+152A_0\dr+(2n+1)\theta_{r-1}
                 +(n+2)\mu(x_{r-1}).
\end{aligned}
\]
By construction of the routing schedule, whenever \(V\) is called,
the ideal row has
\[
                         v_{r-1}(s)=v_{r-1}(d)=0.
\]
Diagonal basis rigidity therefore yields
\[
 (1-\beta)\mu(x_{r-1})
 \le\operatorname{dec}(x_{r-1})_s+
      \operatorname{dec}(x_{r-1})_d
 \le E_{r-1},
\]
and hence \(\mu(x_{r-1})\le2E_{r-1}\le2S_{r-1}\).  Thus the two raw
bounds become
\[
 S_r\le(1+3A_0)S_{r-1}+123A_0\dr,
 \qquad
 S_r\le(4n+6)S_{r-1}+152A_0\dr.
\]
Since \(\kappa=512A_0\), both imply
\[
 S_r\le(\kappa+10)S_{r-1}+10\kappa\dr .
\]
For the calibrated starting root, the identity constraints give
$E_0\le n\dr$, while initial phase calibration gives
$\theta_0\le\lambda$.  Hence
\[
 S_0\le(n+2L)\dr\le\kappa\dr .
\]
Induction for
\(r\le g_0\) gives
\[
 S_r\le11(\kappa+10)^{r+1}\dr\le C_0\dr.
\]
The final accept coordinate is one of the decoded payload coordinates, so
this proves \eqref{eq:word-transfer}.
\end{proof}

\begin{lemma}[Boolean assignment and clause extraction]
\label{lem:boolean-clause-extraction}
Round each \(\widehat b_i\) to its nearer endpoint
\(\widetilde b_i\in\{0,1\}\).  If all declared responses have defect
\(\dr\le\min\{1/\kappa,1/(26C_0)\}\), then
\[
 |\widehat b_i-\widetilde b_i|\le4C_0\dr .
\]
If \(\widetilde b\) falsifies a clause, that clause word has response at
least \(1-13C_0\dr\ge1/2\), whereas its zero target bounds the response by
$\dr<1/2$.
\end{lemma}
\begin{proof}
By \eqref{eq:word-transfer}, the Booleanity response is within
\(C_0\dr\) of \(\widehat b_i(1-\widehat b_i)\).  Since its target is
zero,
\[
 \widehat b_i(1-\widehat b_i)
 \le(C_0+1)\dr\le2C_0\dr.
\]
Now \(\min\{b,1-b\}\le2b(1-b)\) gives the rounding bound.
Complementation preserves the coordinatewise rounding error.  The product of three numbers in
\([0,1]\) is \(1\)-Lipschitz in each coordinate, so a falsified rounded
clause makes \eqref{eq:clause-product} at least
\(1-12C_0\dr\).  Equation \eqref{eq:word-transfer} loses one further
\(C_0\dr\).
\end{proof}

\begin{theorem}[Robust Boolean reduction]
\label{thm:independent-robust-boolean}
There are absolute constants \(\gamma,c>0\) such that the preceding
polynomial reduction maps a bounded-occurrence \(3\)-SAT instance of size
\(N_{\rm bool}\) to a finite partial specification instance satisfying:
\begin{enumerate}[label=(\roman*)]
\item if satisfiable, it has an exact \(K\)-state shared realization whose
fully declared routing maps have \(0\)--\(1\) entries and whose final output
vector is Boolean;
\item if unsatisfiable, the output instance satisfies
\(J_K(I)\ge\gamma N_{\rm bool}^{-c}\);
\item on both sides, every declared fully calibrated specification and the
query-specific static union have exact minimum \(K\).
\end{enumerate}
\end{theorem}
\begin{proof}
A satisfying assignment inserted in \eqref{eq:boolean-V} makes every
Booleanity and clause test zero.  Suppose now that the formula is
unsatisfiable, and put
\[
 \epsilon_{\rm out}:=\min\!\left\{
 \frac1{K_0(2+8L_Q)},\,\frac1\kappa,\,\frac1{26C_0}\right\}.
\]
If $J_K(I)<\epsilon_{\rm out}$, compactness and continuity of the finite partial specification
parameter space give a minimizing realization of some dimension $d\le K$
with defect $\dr=J_K(I)<\epsilon_{\rm out}$.  The first strict inequality and
Lemma~\ref{lem:approx-sector-allocation} force $d=K$ and the saturated sector
allocation.  The other two strict inequalities allow
Lemmas~\ref{lem:robust-boolean-gate} and
\ref{lem:boolean-clause-extraction} to extract a satisfying assignment, a
contradiction.  Therefore $J_K(I)\ge\epsilon_{\rm out}$.
Because \(K_0=O(N_{\rm bool})\), \(L_Q=O(\log N_{\rm bool})\), and \(g_0=6\), the
right side is at least \(\gamma N_{\rm bool}^{-c}\) for some absolute
\(\gamma,c>0\).  The output contains a constant-size test for every input
variable and clause and has polynomial compiler expansion, so its
encoding length is polynomially related to $N_{\rm bool}$ in both directions.
Consequently the reduction may output the rational threshold
\(\epsilon=\epsilon_{\rm out}\), whose encoding length is polynomial and which
satisfies \(\epsilon\ge N_{\rm out}^{-c'}\) for an absolute \(c'>0\).
Item (iii) is
Corollary~\ref{cor:identity-hard-family-preservation}; locally, the
assignment may be chosen separately for each Booleanity or clause specification.
\end{proof}

\begin{proof}[Proofs of Theorems~\ref{thm:finite-partial-specification-exact}
and~\ref{thm:finite-partial-specification-robust}]
\emph{Exact classification.}
Lemma~\ref{lem:finite-partial-specification-etr-membership} gives exact membership, including
the Boolean-final-vector restriction, and Theorem~\ref{thm:exact-gate-hardness}
gives exact \(\exists\mathbb{R}\)-hardness.  In the forward direction, the
reduction produces a Boolean final output vector, whereas its converse is sound
over the full terminal-effect domain.

\emph{Robust classification.}
Lemma~\ref{lem:finite-partial-specification-np-membership} gives robust membership in both
effect variants.  For
hardness, take an arbitrary $\Pi\in\mathsf{PromiseNP}$ and its polynomial-time
verifier.  Cook--Levin and the standard $3$-CNF conversion produce a formula
satisfiable exactly when an accepting certificate exists
~\citep{cook1971complexity}; Tovey's satisfiability-preserving transformation
then produces a $(3,4)$-SAT formula~\citep{tovey1984simplified}.  On promised
YES and NO inputs, Theorem~\ref{thm:independent-robust-boolean} gives,
respectively, zero defect and the inverse-polynomial rational gap above.  This
is a promise-preserving reduction.  Its YES witness has a Boolean terminal
effect, and its soundness gap ranges over the full terminal-effect domain.

\emph{Common hardness-output regime.}
The parser gives \textup{(H1)}, and
Corollary~\ref{cor:identity-hard-family-preservation} gives
\textup{(H2)}--\textup{(H3)} for both source families.  The common
calibration/routing declarations and coincident constraint tests have
identical targets, giving \textup{(H4)}.  Both source constructions use a fixed
finite set of rational constants, and the compiler adds binary calibration,
phase, routing, and basis-identification targets. Hence every hardness target has
constant bit length.
\end{proof}


\end{document}